\documentclass[10pt,twocolumn,twoside]{IEEEtran}

\usepackage{amsmath,graphicx}
\usepackage[english]{babel}
\usepackage[latin1]{inputenc}
\usepackage{tikz}

\newcommand{\displaycomments}

\ifdefined\displaycomments
\newcommand{\note}[1]{\textcolor{red}{\emph{#1}}}
\else
\newcommand{\note}[1]{}
\fi

\newcommand{\displayold}

\ifdefined\displayold
\newcommand{\old}[1]{\textcolor{blue}{{#1}}}
\else
\newcommand{\old}[1]{}
\fi

\newcommand{\eye}{I}

\newcommand{\zeros}{\mathrm{O}}

\usepackage{xspace}
\usepackage{amsmath}
\usepackage{amsthm}
\usepackage{amssymb}
\usepackage{graphicx}      
\usepackage{algorithmicx} 
\usepackage{cite} 

\usetikzlibrary{arrows}
\usepackage{pgfplots}

\DeclareMathOperator*{\argmin}{argmin}

\newcommand{\iverson}[1]{[\![#1]\!]}
\newcommand{\E}{\mathop{\mathbb E}}

\renewcommand{\t}[1]{\mathrm{T}#1}
\newcommand{\N}{\mathcal{N}}

\renewcommand{\baselinestretch}{0.9} 
\usepackage{cite} 
\usepackage{stfloats} 
\usepackage{amsmath}

\def\z{{\mathbf z}}

\renewcommand{\d}[1]{\;\mathrm{d}#1}
\renewcommand{\t}[1]{\mathrm{T}#1}

\newcommand{\tr}{\operatorname{tr}}

\newcommand{\Var}{\operatorname{Var}}
\newcommand{\var}{\operatorname{Var}}

\newtheorem{lemma}{Lemma}

\newcommand{\ST}{\operatorname{ST}}

\newcommand{\G}{\mathcal{G}}

\newcommand{\pluseq}{\stackrel{+}{=}}

\renewcommand{\d}[1]{\;\mathrm{d}#1}
\renewcommand{\t}[1]{\mathrm{T}#1}

\usepackage{booktabs}
\usepackage {ftnxtra}
\usepackage{algpseudocode}
\usepackage{epstopdf}

\newcommand{\figs}{figures/}

\usepackage{xspace}
\newcommand{\STRTVBF}{STF\xspace}
\newcommand{\STRTVBS}{STS\xspace}

\title{ Skew-$t$ Inference with\\ Improved Covariance Matrix Approximation}

\author{Henri Nurminen, Tohid Ardeshiri, Robert Pich\'e, and Fredrik Gustafsson,
\thanks{H. Nurminen and R. Pich\'e are with the Department of Automation Science and Engineering, Tampere University of Technology (TUT), PO Box 692, 33101 Tampere, Finland  (e-mails: henri.nurminen@tut.fi, robert.piche@tut.fi). H. Nurminen receives funding from TUT Graduate School, the Foundation of Nokia Corporation, and Tekniikan edist\"amiss\"a\"ati\"o.}
\thanks{T. Ardeshiri and F. Gustafsson are with the Department of Electrical Engineering, Link\"{o}ping University, 58183 Link\"{o}ping, Sweden, (e-mails: tohid@isy.liu.se, fredrik@isy.liu.se). T. Ardeshiri receives funding from  Swedish research council (VR), project scalable Kalman filters.}
}

\begin{document}
\def\figurename{Fig.}
\maketitle

\begin{abstract}
Filtering and smoothing algorithms for linear discrete-time state-space models with skew-$t$ distributed measurement noise are presented. The proposed algorithms improve upon our earlier proposed filter and smoother using the mean field variational Bayes approximation of the posterior distribution to a skew-$t$ likelihood and normal prior. Our simulations show that the proposed variational Bayes approximation gives a more accurate approximation of the posterior covariance matrix than our earlier proposed method. Furthermore, the novel filter and smoother outperform our earlier proposed methods and conventional low complexity alternatives in accuracy and speed.
\end{abstract}



\section{Introduction} \label{sec:introduction}

Asymmetric and heavy-tailed noise processes are present in many inference problems. In radio signal based distance estimation \cite{GusGun2005,BorsenChen2009,Kok2015}, for example, obstacles cause large positive errors that dominate over symmetrically distributed errors from other sources \cite{kaemarungsi2012}. The skew $t$-distribution \cite{branco2001,azzalini2003, gupta2003skew} is the generalization of the $t$-distribution that has the modeling flexibility to capture both skewness and heavy-tailedness of such noise processes.  
To exemplify this, Fig.~\ref{fig:2Diidcontours} illustrates the contours of the likelihood function for three independent range measurements where some of the measurements are positive outliers. In this example, skew-$t$, $t$, and normal likelihoods are compared. The skew-$t$ likelihood gives a more realistic spread of the probability mass than the normal and $t$ likelihoods. 
Filtering and smoothing algorithms for linear discrete-time state-space models with skew-$t$ measurement noise using a variational Bayes (VB) method are presented in~\cite{nurminen2015a}. This filter is applied to indoor localization with real ultra-wideband data in \cite{nurminen2015b}. 
This letter proposes improvements to the filter and smoother proposed in \cite{nurminen2015a}. 
Analogous to \cite{nurminen2015a}, the measurement noise is modeled by the skew $t$-distribution, and the proposed filter and smoother use a VB approximation of the posterior. However, the main contributions of this letter are (1) a new factorization of the approximate posterior distribution, (2) the application of an existing method for approximating the statistics of a truncated multivariate normal distribution (TMND), and (3) a proof of optimality for a truncation ordering in approximation of the moments of the TMND. A TMND is a multivariate normal distribution whose support is restricted (truncated) by linear constraints and that is re-normalized to integrate to unity. The aforementioned contributions improve the estimation performance by reducing the covariance underestimation common to most VB inference algorithms~\cite[Chapter 10]{Bishop2007}.
 To our knowledge, VB approximations have been applied to the skew $t$-distribution only in our work \cite{nurminen2015a, nurminen2015b} and by Wand et al.\ \cite{wand2011}.


\begin{figure}
\centering
\includegraphics[width=0.65\columnwidth]{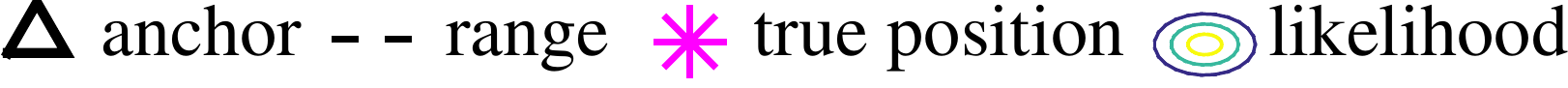}\\\vspace{1mm}
\includegraphics[trim=0mm 0mm 0mm 0mm, clip, width=0.27\columnwidth]{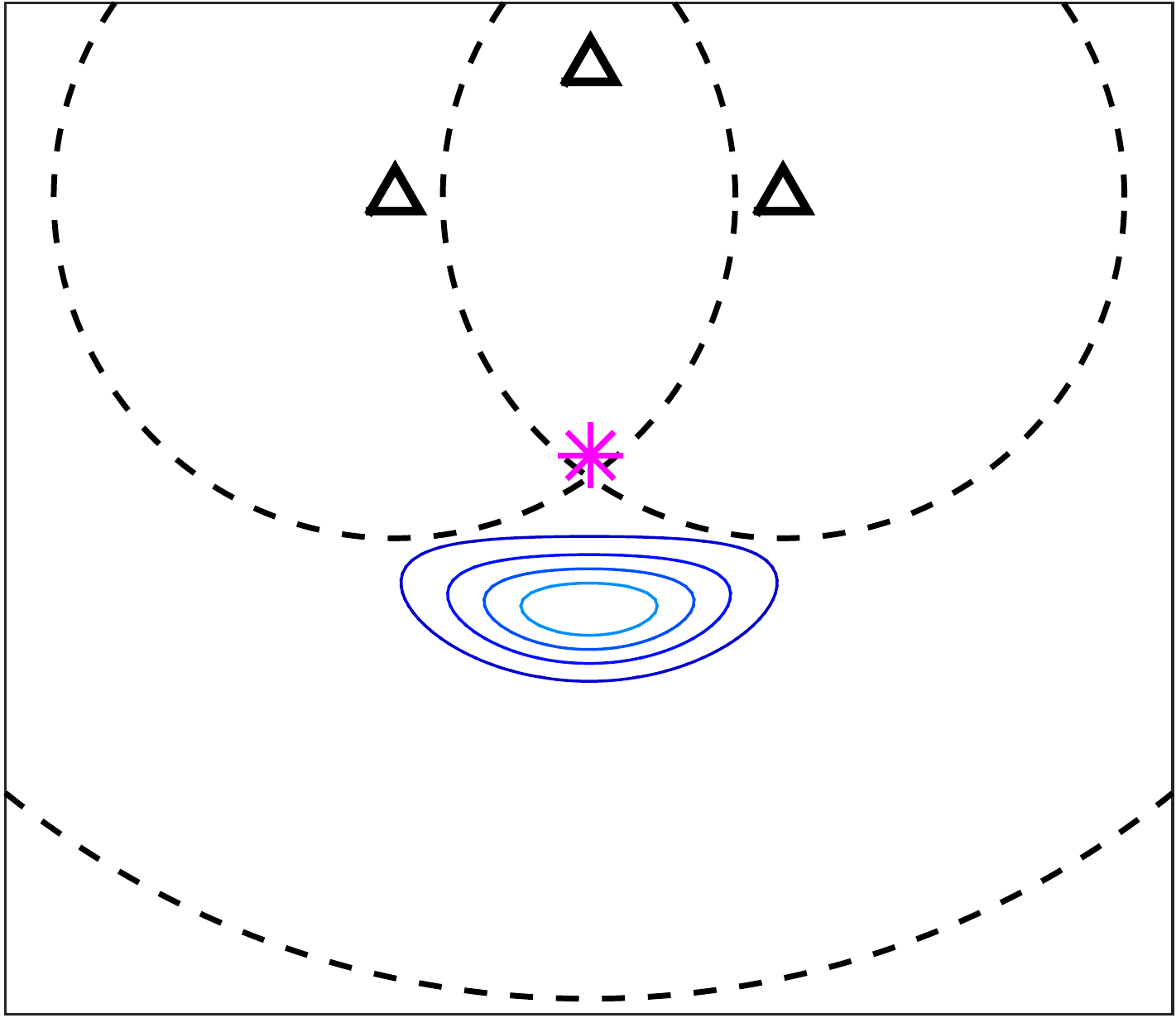}
\includegraphics[trim=0mm 0mm 0mm 0mm, clip, width=0.27\columnwidth]{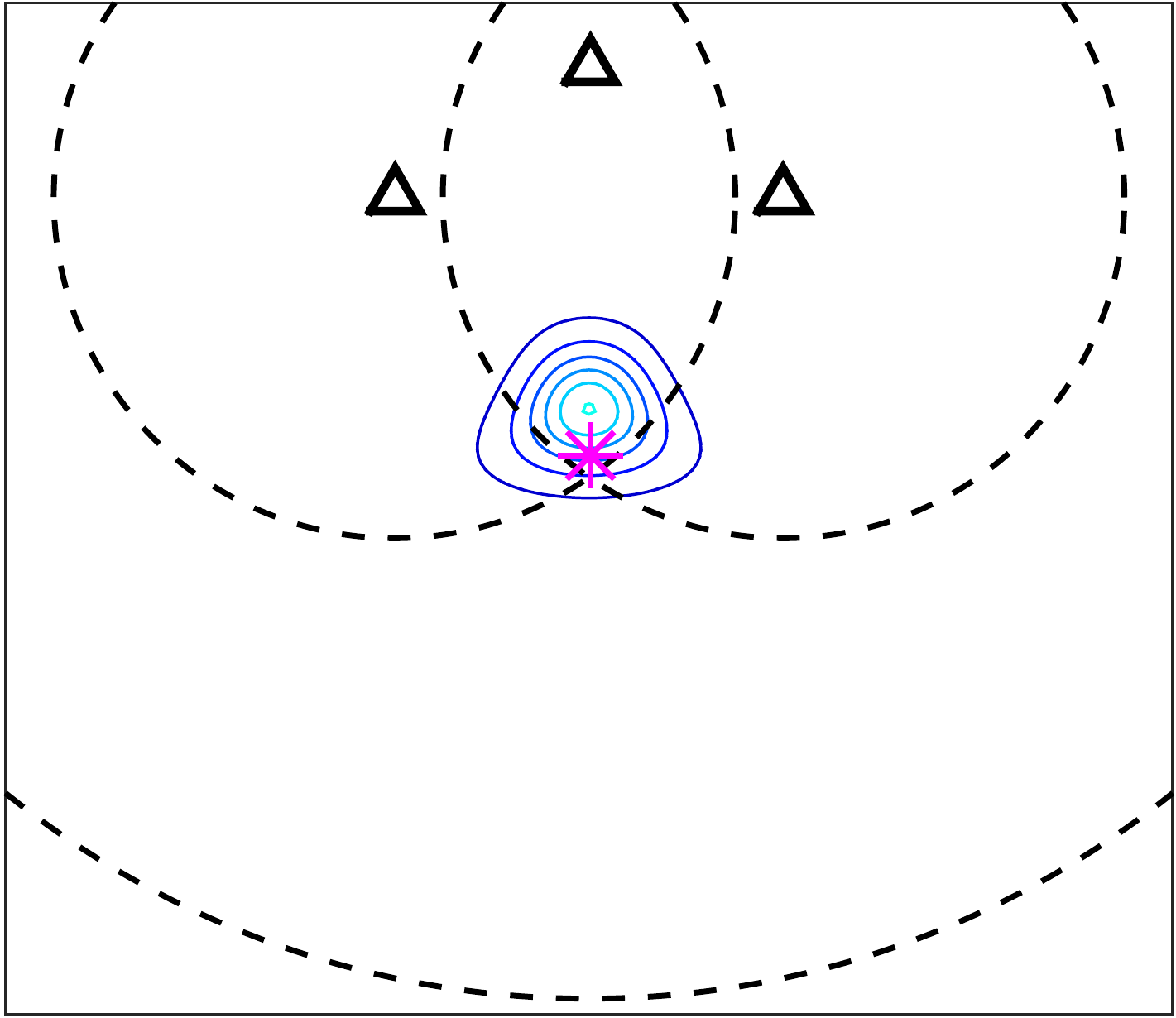}
\includegraphics[trim=0mm 0mm 0mm 0mm, clip, width=0.27\columnwidth]{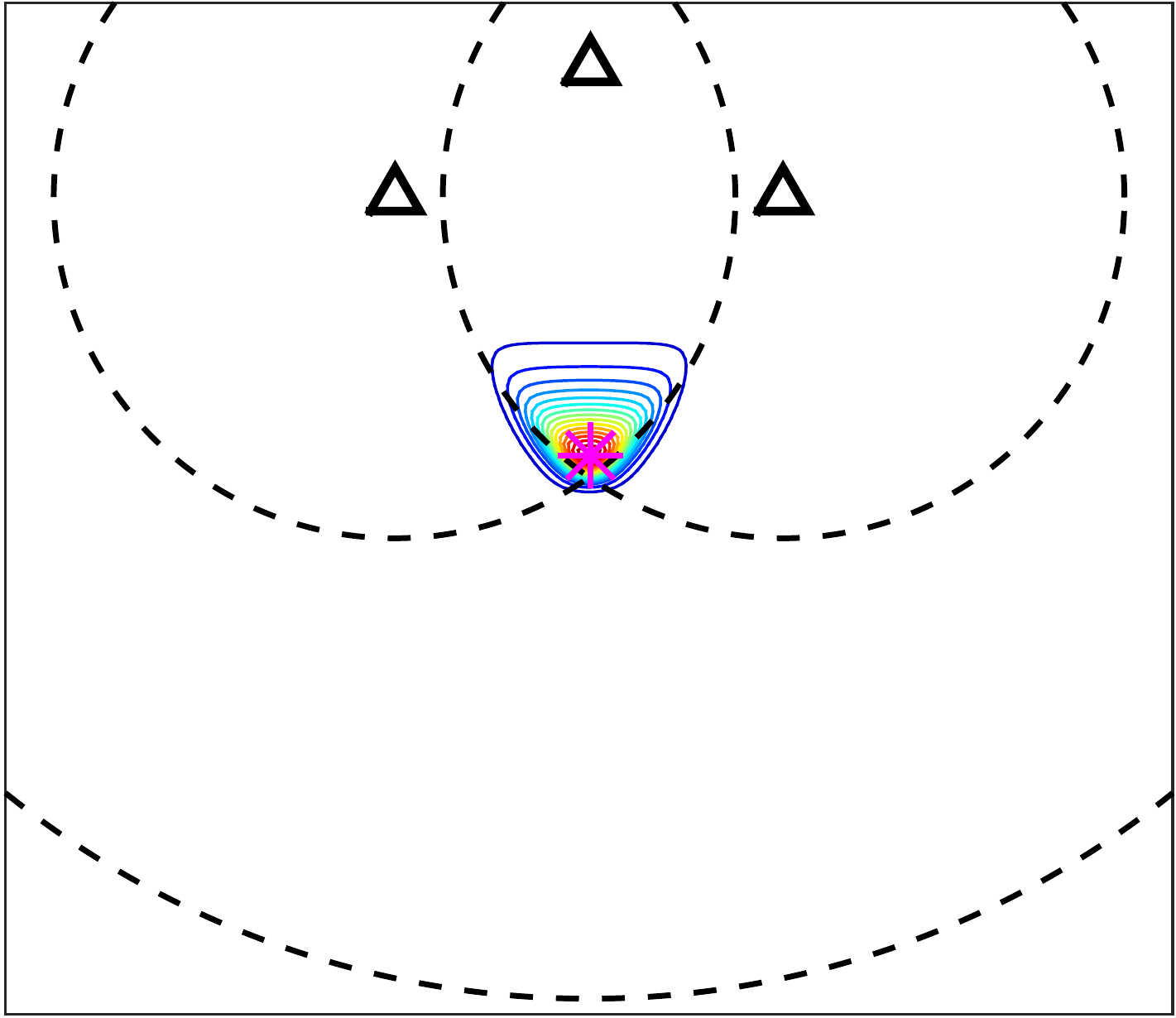}

\vspace{2mm}

\includegraphics[trim=0mm 0mm 0mm 0mm, clip, width=0.27\columnwidth]{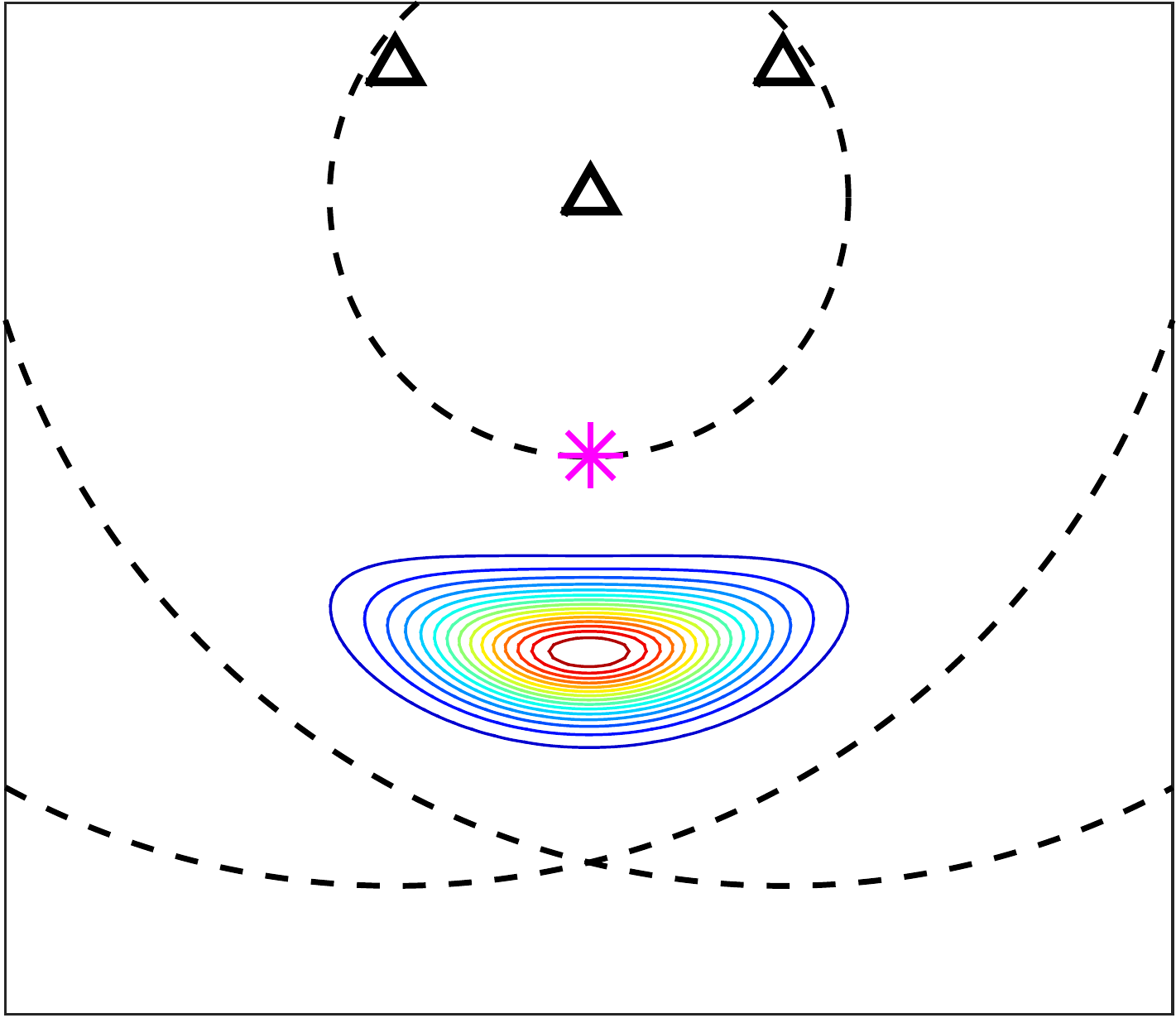}
\includegraphics[trim=0mm 0mm 0mm 0mm, clip, width=0.27\columnwidth]{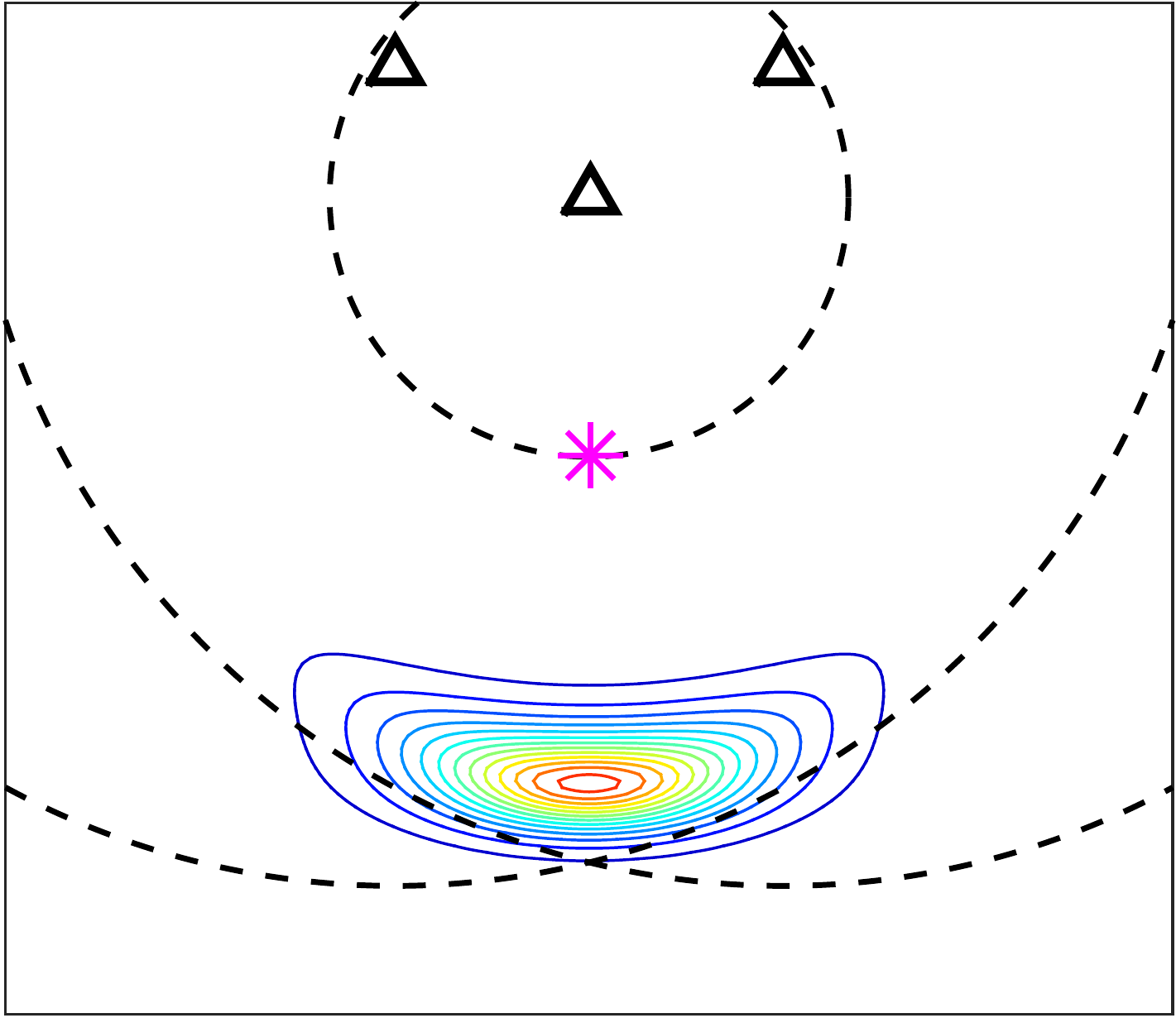}
\includegraphics[trim=0mm 0mm 0mm 0mm, clip, width=0.27\columnwidth]{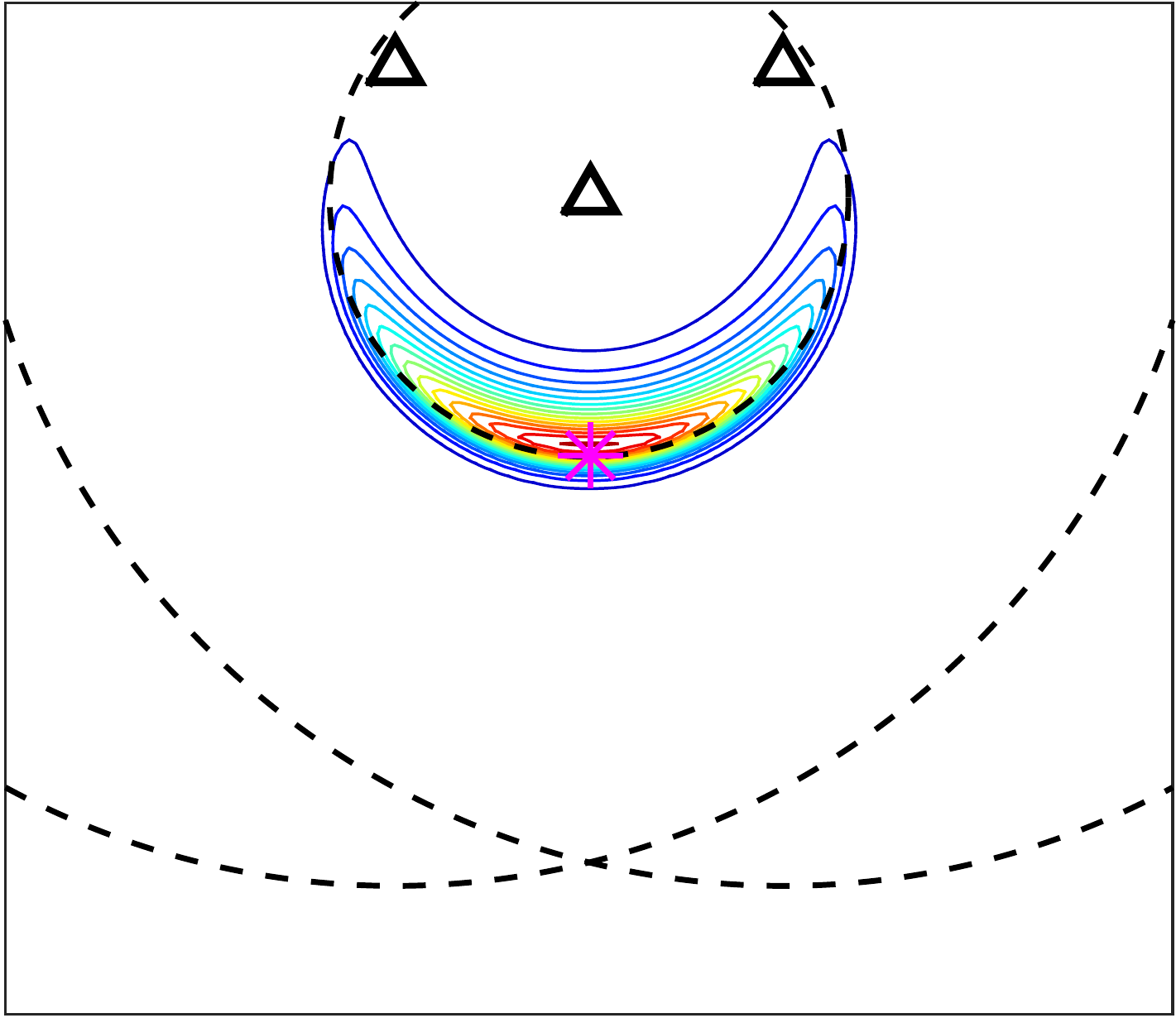}
\caption{The contours of the likelihood function for three range measurements for the normal (left), $t$ (middle) and skew-$t$ (right) measurement noise models are presented. The $t$ and skew-$t$ likelihoods can handle one outlier (upper row), while only the skew-$t$ model can handle the two positive outlier measurements (bottom row) due to its asymmetry. The likelihoods' parameters are selected such that the first two moments of the normal, $t$ and skew-$t$ PDFs coincide. } \label{fig:2Diidcontours}
\end{figure}

\section{Problem formulation} \label{sec:problem_formulation}
Consider the linear and Gaussian state evolution model
\begin{subequations}
\label{eq:state-evo}
\begin{align}
x_{k+1}&=\ Ax_k+w_k, & w_k&\stackrel{\text{iid}}{\sim}\N(0,Q),\\
p(x_1)&=\N(x_1;x_{1|0},P_{1|0}),\label{eq:prior}
\end{align}
\end{subequations}
where 
 $\N(\cdot;\mu,\Sigma)$ denotes a (multivariate) normal PDF with mean $\mu$ and covariance matrix $\Sigma$;
 $A \in \mathbb{R}^{n_x\times n_x}$ is the state transition matrix;
 $x_k\in\mathbb{R}^{n_x}$ indexed by $1\!\leq\! k\!\leq\! K$ is the state to be estimated with initial prior distribution \eqref{eq:prior}, 
 where the subscript ``$a|b$" is read ``at time $a$ using measurements up to time $b$". 
Further, consider the measurements $y_k\in\mathbb{R}^{n_y}$ to be governed by the measurement equation
\begin{align}
y_{k}&= C x_k+e_k, & [e_k]_i&\stackrel{\text{iid}}{\sim} \ST(0,R_{ii},\Delta_{ii},\nu_i), \label{eq:measurementmodel}
\end{align}
where the measurement noise distribution is a product of independent univariate skew $t$-distributions. This  model is justified in applications where one-dimensional data from different sensors can be assumed to have statistically independent noise~\cite{nurminen2015b}. The PDF and the first two moments of the skew $t$-distribution can be found in~\cite{nurminen2015b} and \cite{sahu2009}, respectively.

The model~\eqref{eq:measurementmodel} admits the hierarchical representation 
\begin{subequations}
\label{eq:hierarchical}
\begin{align}
y_k|x_k,u_k,\Lambda_k &\thicksim \N(Cx_k+\Delta u_k,\Lambda_k^{-1}R),\\
u_k|\Lambda_k &\thicksim \N_+(0,\Lambda_k^{-1}),\\
[\Lambda_k]_{ii} &\thicksim \G(\tfrac{\nu_i}{2},\tfrac{\nu_i}{2}),
\end{align}
\end{subequations}
where
 $R\in\mathbb{R}^{n_y\times n_y}$ is a diagonal matrix of which the square roots of the diagonal elements, $\sqrt{R_{ii}}$, are the spread parameters of the skew $t$-distribution in
 \eqref{eq:measurementmodel};
 $\Delta\in\mathbb{R}^{n_y\times n_y}$ is a diagonal matrix whose diagonal elements $\Delta_{ii}$ are the shape parameters;
 $\nu\in\mathbb{R}^{n_y}$ is a vector whose elements $\nu_i$ are the degrees of freedom;
 $C \in \mathbb{R}^{n_y\times n_x}$ is the measurement matrix;
 $\{w_k\in\mathbb{R}^{n_x}| 1\!\leq\! k \!\leq\! K\}$  and $\{e_k\in\mathbb{R}^{n_y}| 1\!\leq\! k \!\leq\! K\}$  are mutually independent noise sequences;  the operator $[\cdot]_{ij}$ gives the $(i,j)$ entry of its argument;
 $\Lambda_k$ is a diagonal matrix with a priori independent random diagonal elements $[\Lambda_k]_{ii}$. 
 Also, $\N_+(\mu,\Sigma)$ is the TMND with closed positive orthant as support, location parameter $\mu$, and squared-scale matrix $\Sigma$. Furthermore, $\G(\alpha, \beta)$ is the gamma distribution with shape parameter $\alpha$ and rate parameter $\beta$. Models where the measurement noise components are vector-valued with independently multivariate skew-$t$ distributed noises \cite{branco2001,azzalini2003, gupta2003skew,sahu2003,lin2010,lee2013} require only a straightforward modification to the update of the approximate posterior of $\Lambda_k$ in the proposed filtering and smoothing algorithms.

Bayesian smoothing means finding the smoothing posterior $p(x_{1:K},u_{1:K},\Lambda_{1:K}|y_{1:K})$. 
In \cite{nurminen2015a}, the smoothing posterior  is approximated by a factorized distribution of the form $q_{\text{\cite{nurminen2015a}}}\!\triangleq\! q_x(x_{1:K})q_u(u_{1:K})q_\Lambda(\Lambda_{1:K})$. Subsequently, the approximate posterior distributions are computed using the VB approach. The VB approach minimizes the Kullback--Leibler divergence (KLD) $D_{\text{KL}}(q||p)\!\triangleq\!\int q(x)\log\frac{q(x)}{p(x)}\!\d x$~\cite{CoverT2006} of the true posterior from the factorized approximation. 
That is, $D_{\text{KL}}(q_{\text{\cite{nurminen2015a}}}||p(x_{1:K},u_{1:K},\Lambda_{1:K}|y_{1:K}))$ is minimized in \cite{nurminen2015a}. 

The numerical simulations in \cite{nurminen2015a} manifest the covariance underestimation of the VB approach, which is a known weakness of the method \cite[Chapter 10]{Bishop2007}.    
The aim of this letter is to reduce the covariance underestimation of the filter and smoother proposed in \cite{nurminen2015a} by removing independence approximations of the posterior approximation.


\section{Proposed Solution} \label{sec:solution}
Using Bayes' theorem, the state evolution model~\eqref{eq:state-evo}, and the likelihood \eqref{eq:hierarchical},  the joint smoothing posterior PDF can be derived as in~\cite{nurminen2015a}. 
This posterior is not analytically tractable. We propose to seek an approximation in the form
\begin{align}
\label{eq:factors}
p(x_{1:K},&u_{1:K},\Lambda_{1:K}|y_{1:K})\approx q_{xu}(x_{1:K},u_{1:K})\,q_\Lambda(\Lambda_{1:K}) ,
\end{align}
where the factors in \eqref{eq:factors} are specified by
\begin{align}
&\hat{q}_{xu},\hat{q}_{\Lambda}=\argmin_{{q}_{xu},{q}_{\Lambda}}D_{\text{KL}}(q_\text{N}||p(x_{1:K},u_{1:K},\Lambda_{1:K}|y_{1:K}))\nonumber
\end{align}
and where $q_\text{N}\!\triangleq\! q_{xu}(x_{1:K},u_{1:K})q_\Lambda(\Lambda_{1:K})$. Hence, $x_{1:K}$ and $u_{1:K}$ are not approximated as independent as in \cite{nurminen2015a} because they can be highly correlated a posteriori \cite{nurminen2015a}. The analytical solutions for $\hat{q}_{xu}$ and $\hat{q}_\Lambda$ are obtained by cyclic iteration of
\small{
\begin{subequations}
\label{eqn:IterativeOptimization}
\begin{align}
\log {q}_{xu}(\cdot) &\leftarrow \E_{{q}_{\Lambda}}[\log p(y_{1:K},x_{1:K},u_{1:K},\Lambda_{1:K})]+c_{xu}\label{eqn:IterativeOptimizationxu}\\
\log {q}_{\Lambda}(\cdot) &\leftarrow \E_{{q}_{xu}}[\log p(y_{1:K},x_{1:K},u_{1:K},\Lambda_{1:K})]+c_{\Lambda}\label{eqn:IterativeOptimizationL}
\end{align}
\end{subequations}}\normalsize
where the expected values on the right hand sides are taken with respect to the current $q_{xu}$ and $q_\Lambda$~\cite[Chapter 10]{Bishop2007}\cite{TzikasLG2008,Beal03}. Also, $c_{xu}$ and $c_\Lambda$  are constants with respect to the variables $(x_{1:K},u_{1:K})$ and $\Lambda_{1:K}$,  respectively.


Computation of the expectation in~\eqref{eqn:IterativeOptimizationL} requires the first two moments of a TMND, because the support of $u_{1:K}$ is the non-negative orthant.
  These moments can be computed using the formulas presented in \cite{tallis1961}. They require evaluating the CDF (cumulative distribution function) of general multivariate normal distributions. The \textsc{Matlab} function \verb+mvncdf+ implements the numerical quadrature  of \cite{genz2004} in $2$ and $3$ dimensional cases and the quasi-Monte Carlo method of \cite{Genz02} for the dimensionalities $4-25$. However, these methods can be prohibitively slow. Therefore, we approximate the TMND's moments using the fast recursive algorithm suggested in \cite{perala2008,simon2010}. The method is initialized with the original normal density whose parameters are then updated by applying one linear constraint at a time. For each constraint, the mean and covariance matrix of the once-truncated normal distribution are computed analytically, and the once-truncated distribution is approximated by a non-truncated normal with the updated moments.
 
The result of the recursive truncation depends on the order in which the constraints are applied. Finding the optimal order of applying the truncations is a combinatorial problem. Hence, we choose a greedy approach, 
whereby the constraint to be applied is chosen from among the remaining constraints so that the resulting once-truncated normal is closest to the true TMND. By Lemma~\ref{lem:optimal-seq}, the optimal constraint in KLD-sense is the one that truncates the most probability. The obtained algorithm with the optimal processing sequence for computing the mean and covariance of a given normal distribution truncated to the positive orthant is given in Table~\ref{table:recursive}.
\begin{lemma} \label{lem:optimal-seq}
Let $p(\z)$ be a {TMND} with the support $\{\z \geq 0\}$ and $q(\z)=\N(\z;\mu,\Sigma)$. Then,
\begin{align}
\argmin_i D_{\normalfont{\text{KL}}}\left(p(\z)\,\big| \big|\,\tfrac{1}{c_i}q(\z)\iverson{\z_i \geq 0}\right) = \argmin_i \tfrac{\mu_i}{\sqrt{\Sigma_{ii}}} ,
\end{align} 
where $\mu_i$ is the $i$th element of $\mu$, $\Sigma_{ii}$ is the $i$th diagonal element of $\Sigma$, $\iverson{\cdot}$ is the Iverson bracket, and $c_i\!=\! \int{q(\z)\iverson{\z_i \geq 0}}\d\z$.
\begin{align*}
\text{Proof: }& D_{\normalfont{\text{KL}}}\left(p(\z)\,\big| \big|\,\tfrac{1}{c_i}q(\z)\iverson{\z_i \geq 0}\right)\nonumber\\
&\pluseq -\int_0^\infty{p(\z)\log(\tfrac{1}{c_i}q(\z)\iverson{\z_i \geq 0})\d\z}\nonumber\\
&= \log{c_i} - \int_0^\infty{p(\z)\log q(\z)\d\z} - 1 \pluseq \log{c_i} ,
\end{align*} 
where $\pluseq$ means equality up to an additive constant. 
Since  $c_i$ is an increasing function of $\frac{\mu_i}{\sqrt{\Sigma_{ii}}}$ the proof follows. \vspace{-2ex} \begin{flushright}$\blacksquare$\end{flushright}
\end{lemma}
%

\begin{table}
\caption{Optimal Recursive Truncation to the Positive Orthant}\label{table:recursive}
\vspace{-5mm}\rule{\columnwidth}{1pt}
\begin{algorithmic}[1]
\State \textbf{Inputs:} $\mu$, $\Sigma$, and the set of the truncated components' indices $\mathcal{T}$
\While{$\mathcal{T} \neq \emptyset$}
	\State $k \gets \argmin_i\{ \mu_i/\sqrt{\Sigma_{ii}} \mid i \in \mathcal{T} \}$
	\State $\xi \gets \mu_k / \sqrt{\Sigma_{kk}}$
	\If{$\Phi(\xi)$ does not underflow to $0$}
		\State $\epsilon \gets \phi(\xi) / \Phi(\xi)$ \Comment{$\phi$ is the PDF of $\N(0,1)$, $\Phi$ its CDF}
		\State $\mu \gets \mu + (\epsilon/\sqrt{\Sigma_{kk}}) \cdot \Sigma_{:,k}$
		\State $\Sigma \gets \Sigma - ((\xi\epsilon+\epsilon^2)/\Sigma_{kk}) \cdot \Sigma_{:,k}\Sigma_{k,:}$
	\Else
		\State $\mu \gets \mu + (-\xi/\sqrt{\Sigma_{kk}}) \cdot \Sigma_{:,k}$ \Comment{$\lim_{\xi\rightarrow-\infty}(\epsilon+\xi)=0$ \cite{perala2008}}
		\State $\Sigma \gets \Sigma - (1/\Sigma_{kk}) \cdot \Sigma_{:,k}\Sigma_{k,:}$ \Comment{$\lim_{\xi\rightarrow-\infty}(\xi\epsilon+\epsilon^2)=1$ \cite{perala2008}}
	\EndIf
	\State $\mathcal{T} \gets \mathcal{T} \backslash \{k\}$
\EndWhile
\State \textbf{Outputs: $\mu$ and  $\Sigma$;} ($[\mu,\Sigma] \gets \texttt{rec\_trunc}(\mu,\Sigma,\mathcal{T}$))
\end{algorithmic}
\noindent \rule{\columnwidth}{1pt}\vspace{0mm}
\end{table}

The recursion~\eqref{eqn:IterativeOptimization} is convergent to a local optimum~\cite[Chapter 10]{Bishop2007}. However, there is no proof of convergence available when the moments of the TMND are approximated. In spite of lack of a convergence proof the iterations did not diverge in the numerical simulations presented in section~\ref{sec:simulations}. 

The  derivations for the expectations of \eqref{eqn:IterativeOptimization} are presented in the appendixes. In the smoother, the update \eqref{eqn:IterativeOptimizationxu} includes a forward filtering step of the Rauch--Tung--Striebel smoother (RTSS) \cite{RTS-1965} where the first filtering posterior is a TMND. The TMND is approximated as a multivariate normal distribution 
whose parameters are obtained using the recursive truncation. This approximation enables recursive forward filtering and the use of RTSS's backward smoothing step that gives normal approximations to the marginal smoothing posteriors $q_{xu}(x_k, u_k) \!\approx\! \N(\left[\begin{smallmatrix}x_k\\u_k\end{smallmatrix}\right]; z_{k|K},Z_{k|K})$. After the iterations converge, the variables $u_{1:K}$ are integrated out to get the approximate smoothing posteriors $q_x(x_k)\!=\!\N(x_k;x_{k|K},P_{k|K})$, where the parameters $x_{k|K}$ and $P_{k|K}$ are the output of the skew $t$ smoother (STS) algorithm in Table~\ref{table:smoothing}. STS can be restricted to an online recursive algorithm to synthesize a filter which is summarized in Table~\ref{table:filtering}. In the filter, the output of a filtering step is also a TMND which in analogy to STS is approximated by a multivariate normal distribution to have a recursive algorithm. Using recursive truncation, the TMND is approximated by a normal distribution $q_x(x_k)\!=\!\N(x_k;x_{k|k},P_{k|k})$ whose parameters are the outputs of the skew $t$ filter (STF) algorithm in Table~\ref{table:filtering}. 

\begin{table}
\caption{Smoothing for skew-$t$ measurement noise}\label{table:smoothing}
\vspace{-5mm}\rule{\columnwidth}{1pt}
\newcommand{\mlambda}[1]{\Lambda_{#1|K}}
\begin{algorithmic}[1]
\State \textbf{Inputs:} $A$, $C$, $Q$, $R$, $\Delta$,  $\nu$, $x_{1|0}$,  $P_{1|0}$ and $y_{1:K}$
\State $A_z \gets  \left[\begin{smallmatrix}A&0\\0&0\end{smallmatrix}\right]$, $C_z \gets \left[\begin{smallmatrix}C&\Delta\end{smallmatrix}\right]$
\Statex \hspace{0mm}\textit{initialization}
\State $\mlambda{k} \gets I_{n_y}$ for $k=1\cdots K$
\Repeat
\Statex \hspace{2mm}\textit{update $q_{xu}(x_{1:K},u_{1:K})$ given $q_\Lambda(\Lambda_{1:K})$}
	\For{$k$ = 1 to $K$}
		\State $Z_{k|k-1} \gets \mathrm{blockdiag}(P_{k|k-1}, \mlambda{k}^{-1})$
		\State $K_{z} \gets Z_{k|k-1} C_z^\t (CP_{k|k-1}C^\t\!+\!\Delta\mlambda{k}^{-1}\Delta^\t\!+\!\mlambda{k}^{-1}R)^{-1}$
		\State $\widetilde{z}_{k|k} \gets \left[\begin{smallmatrix} x_{k|k-1}\\0 \end{smallmatrix}\right]+K_z(y_k-Cx_{k|k-1})  $
		\State $\widetilde{Z}_{k|k}\gets (I-K_z C_z)P_{k|k-1}$
		\State $[z_{k|k}, Z_{k|k}] \gets \texttt{rec\_trunc}(\widetilde{z}_{k|k}, \widetilde{Z}_{k|k}, \{n_x+1\cdots n_x+n_y\})$
		\State $x_{k|k} \gets [z_{k|k}]_{1:n_x},\ P_{k|k} \gets [Z_{k|k}]_{1:n_x,1:n_x}$
    		\Statex \vspace{-0.9ex}
	 	\State $x_{k+1|k} \gets Ax_{k|k}$
	  	\State $P_{k+1|k} \gets AP_{k|k}A^\t+Q$	
	\EndFor
	\For{$k$ = $K-1$ down to $1$ }
		\State $G_k\gets Z_{k|k} A_z Z_{k+1|k}^{-1}$
		\State $z_{k|K}\gets z_{k|k}+G_k(z_{k+1|K}-A_z z_{k|k})$
		\State $Z_{k|K}\gets Z_{k|k}+G_k(Z_{k+1|K}-Z_{k+1|k})G_k^\t$
		\State $x_{k|K} \gets [z_{k|K}]_{1:n_x},\ P_{k|K} \gets [Z_{k|K}]_{1:n_x,1:n_x}$
		\State $u_{k|K} \!\!\gets\!\! [z_{k|K}]_{n_x+(1:n_y)}, U_{k|K} \!\!\gets\!\! [Z_{k|K}]_{n_x+(1:n_y),n_x+(1:n_y)}$
	\EndFor
	\Statex \hspace{2mm}\textit{update $q_\Lambda(\Lambda_{1:K})$ given $q_{xu}(x_{1:K},u_{1:K})$}
	\For{$k$ = $1$ to $K$}
		\State $\Psi_k \gets (y_k-C_z z_{k|K})(y_k-C_z z_{k|K})^\t R^{-1} + C_z Z_{k|K} C_z^\t R^{-1}$
		\Statex $\hspace{1.5cm}  +u_{k|K}u_{k|K}^\t + U_{k|K}$
		\State $[\mlambda{k}]_{ii} \gets\frac{\nu_i+2}{\nu_i+[\Psi_k]_{ii}}$
	\EndFor	
\Until{\textbf{converged}}
\State \textbf{Outputs: $x_{k|K}$ and  $P_{k|K}$ for $k=1\cdots K$ } 
\end{algorithmic}
\noindent \rule{\columnwidth}{1pt}\vspace{0mm}
\end{table}
\begin{table}[t]
\caption{Filtering for skew-$t$ measurement noise}\label{table:filtering}
\vspace{-5mm}\rule{\columnwidth}{1pt}
\newcommand{\mlambda}[1]{\Lambda_{#1|k}}
\begin{algorithmic}[1]
\State \textbf{Inputs:} $A$, $C$, $Q$, $R$, $\Delta$,  $\nu$, $x_{1|0}$,  $P_{1|0}$ and $y_{1:K}$
\State $C_z \gets \left[\begin{smallmatrix}C&\Delta\end{smallmatrix}\right]$
\For{$k$ = 1 to $K$}
	\Statex \hspace{2mm}\textit{initialization}
	\State $\mlambda{k} \gets I_{n_y}$
	\Repeat
		\Statex \hspace{6mm}\textit{update $q_{xu}(x_k,u_k)=\N(\left[\begin{smallmatrix}x_k\\u_k\end{smallmatrix}\right];z_{k|k},Z_{k|k})$ given $q_\Lambda(\Lambda_k)$}
\State $Z_{k|k-1} \gets \mathrm{blockdiag}(P_{k|k-1}, \mlambda{k}^{-1})$
		\State $K_{z} \gets Z_{k|k-1} C_z^\t (CP_{k|k-1}C^\t\!+\!\Delta\mlambda{k}^{-1}\Delta^\t\!+\!\mlambda{k}^{-1}R)^{-1}$
		\State $\widetilde{z}_{k|k} \gets \left[\begin{smallmatrix} x_{k|k-1}\\0 \end{smallmatrix}\right]+K_z(y_k-Cx_{k|k-1})  $
		\State $\widetilde{Z}_{k|k}\gets (I-K_z C_z)P_{k|k-1}$
		\State $[z_{k|k}, Z_{k|k}] \gets \texttt{rec\_trunc}(\widetilde{z}_{k|k}, \widetilde{Z}_{k|k}, \{n_x+1\cdots n_x+n_y\})$
		\State $x_{k|k} \gets [z_{k|k}]_{1:n_x},\ P_{k|k} \gets [Z_{k|k}]_{1:n_x,1:n_x}$
		\State $u_{k|k} \!\!\gets\!\! [z_{k|k}]_{n_x+(1:n_y)}, U_{k|k} \!\!\gets\!\! [Z_{k|k}]_{n_x+(1:n_y),n_x+(1:n_y)}$
		\Statex \hspace{6mm}\textit{update $q_\Lambda(\Lambda_k)=\prod_{i=1}^{n_y}\G\left([\Lambda_k]_{ii};\frac{\nu_i}{2}+1,\frac{\nu_i+[\Psi_k]_{ii}}{2}\right)$ }
		\Statex \hspace{6mm}\textit{ given $q_{xu}(x_k,u_k)$}
		\State $\Psi_k \gets (y_k-C_z z_{k|k})(y_k-C_z z_{k|k})^\t R^{-1} + C_z Z_{k|k} C_z^\t R^{-1}$
		\Statex $\hspace{1.5cm}  +u_{k|k}u_{k|k}^\t + U_{k|k}$
		\State $[\mlambda{k}]_{ii} \gets\frac{\nu_i+2}{\nu_i+[\Psi_k]_{ii}}$
		\Until{\textbf{converged}}
	\State $x_{k+1|k} \gets Ax_{k|k}$
	\State $P_{k+1|k} \gets AP_{k|k}A^\t+Q$	
\EndFor
\State \textbf{Outputs: $x_{k|k}$ and  $P_{k|k}$ for $k=1\cdots K$ } 
\end{algorithmic}
\noindent \rule{\columnwidth}{1pt}\vspace{0mm}
\end{table}


\section{Simulations} \label{sec:simulations}

Our numerical simulations use satellite navigation pseudorange measurements of the model
\begin{equation} \label{eq:measmodel}
[y_k]_i\mid x_k \sim \mathrm{ST}(\left\| s_i - [x_k]_{1:3} \right\| + [x_k]_4, 1\,\text{m}, \delta\,\text{m}, 4)
\end{equation}
where $s_i$ is the $i$th satellite's position, $[x_k]_4$ is bias with prior $\N(0,(0.75\,\text{m})^2)$, and $\delta$ is skewness parameter. The linearization error is negligible because the satellites are far. The state model is a random walk with process covariance $Q\! =\! \mathrm{diag}( (q\,\text{m})^2, (q\,\text{m})^2, (0.2\,\text{m})^2, 0)$, where $q$ is a parameter. A satellite constellation of Global Positioning System provided by the International GNSS service \cite{dow2009} is used with 8 measured satellites. The RMSE is computed for $[x_k]_{1:3}$.

\subsection{Computation of TMND statistics}

In this subsection we study the computation of the moments of the untruncated components of a TMND. One state and one measurement vector per Monte Carlo replication are generated from the model \eqref{eq:measmodel} with $\nu\!=\!\infty$ degrees of freedom (corresponding to skew-normal likelihood), prior $x\!\sim\!\N(0,\mathrm{diag}(\rho\,\text{m}^2,\rho\,\text{m}^2,(0.22\,\text{m})^2,(0.1\,\text{m})^2))$, and 10\,000 replications. The compared methods are recursive truncations with the optimal truncation order (RTopt) and with random order (RTrand), the variational Bayes (VB), and the analytical formulas of \cite{tallis1961} using \textsc{Matlab} function \verb+mvncdf+ (MVNCDF). In RTrand any of the non-optimal constraints is chosen at each truncation.
VB is an update of the skew $t$ variational Bayes filter (STVBF) \cite{nurminen2015a} where $\overline{\Lambda_k}\!=\!\eye$ and the VB iteration is terminated when the position estimate changes less than 0.005\,m or at the 1000th iteration.

Fig.\ \ref{fig:mvncdftest} shows distributions of the distance from the estimate of the bootstrap particle filter (PF) with 100\,000 samples. The box levels are 5\,\%, 25\,\%, 50\,\%, 75\,\%, and 95\,\% quantiles and the asterisks show minimum and maximum values. With small $\rho$ and $\delta$ the differences between RTrand, RTopt, and MVNCDF are small. With large $\delta$ there are statistically significant differences as the $p$-values of two-sided Wilcoxon signed rank test in Fig.\ \ref{fig:mvncdftest} show. 
RTopt outperforms RTrand in the cases with high skewness, which reflects the result of Lemma~\ref{lem:optimal-seq}. MVNCDF is more accurate than RTopt in the cases with high skewness, but MVNCDF's computational load is roughly 40\,000 times that of the RTopt. This justifies the use of recursive truncation approximation.

The approximation of the posterior covariance matrix is tested by studying the normalized estimation error squared (NEES) values \cite[Ch.\ 5.4.2]{bar-shalom} shown by Fig. \ref{fig:neestest}. If the covariance matrix is correct, the expected value of NEES is the state dimensionality 3 \cite[Ch.\ 5.4.2]{bar-shalom}. VB gets large NEES values when $\delta$ is large, which indicates that VB underestimates the covariance matrix. RTopt and RTrand give NEES values closest to 3, so the recursive truncation provides the most accurate covariance matrix approximation.

\begin{figure}[t]
\includegraphics[width=0.537 \columnwidth]{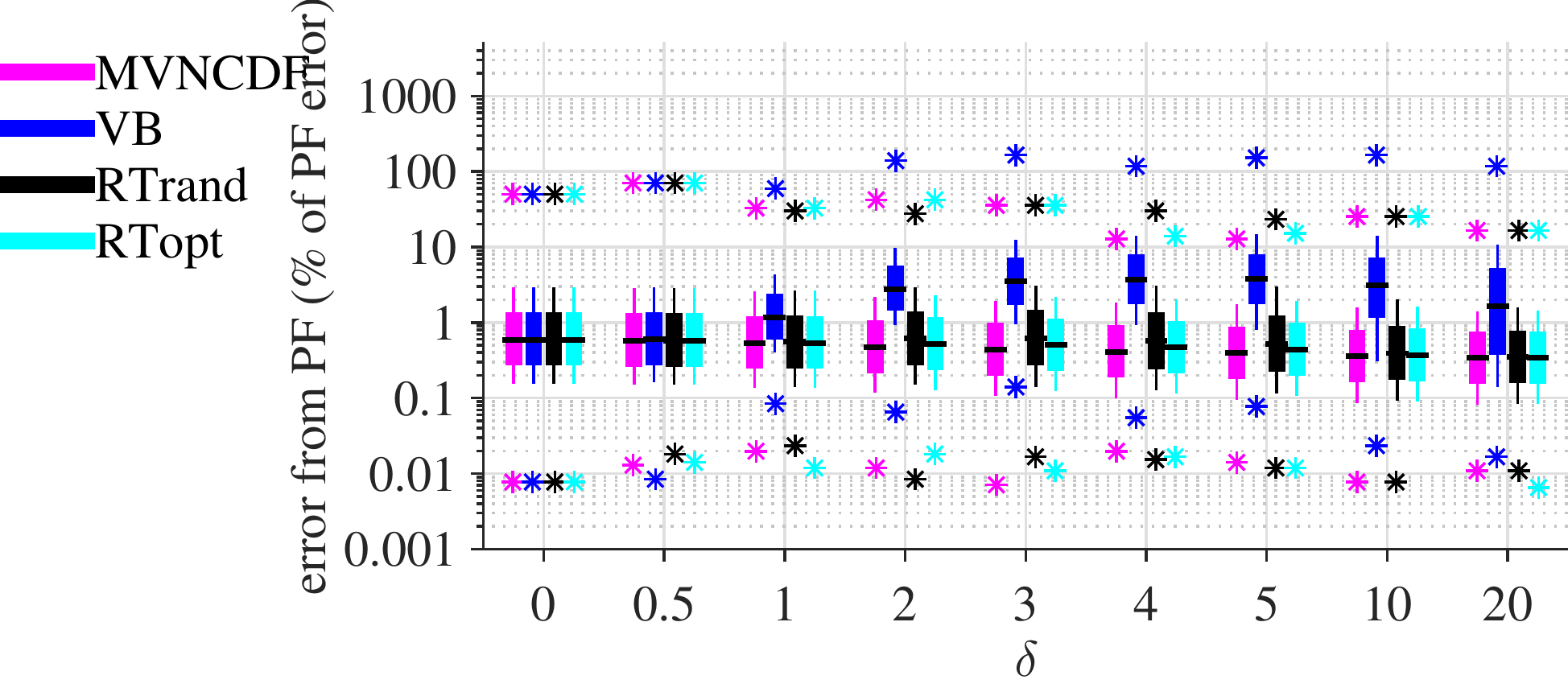}
\hfil
\includegraphics[width=0.437 \columnwidth]{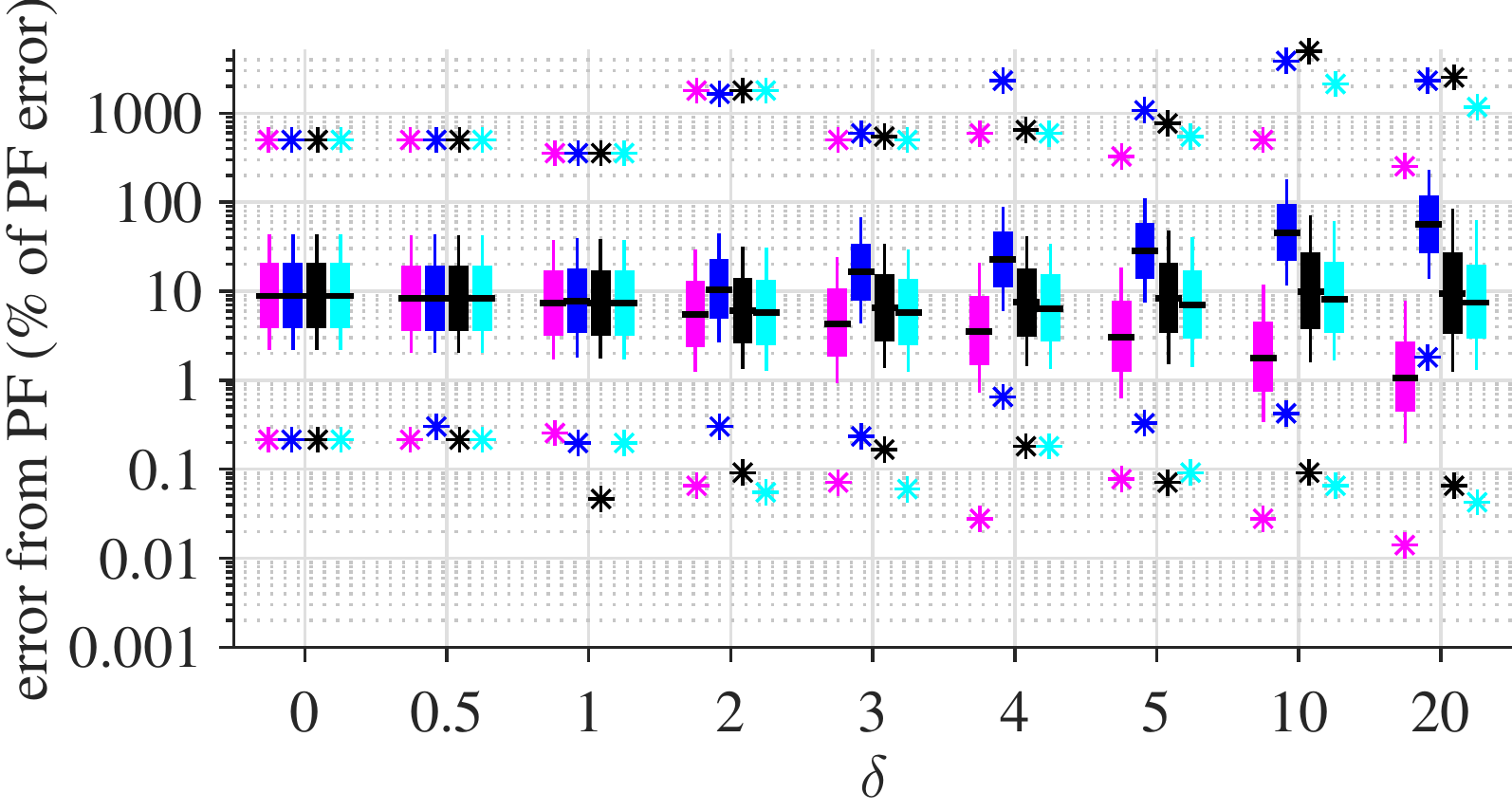}
\\
\includegraphics[width=0.537 \columnwidth]{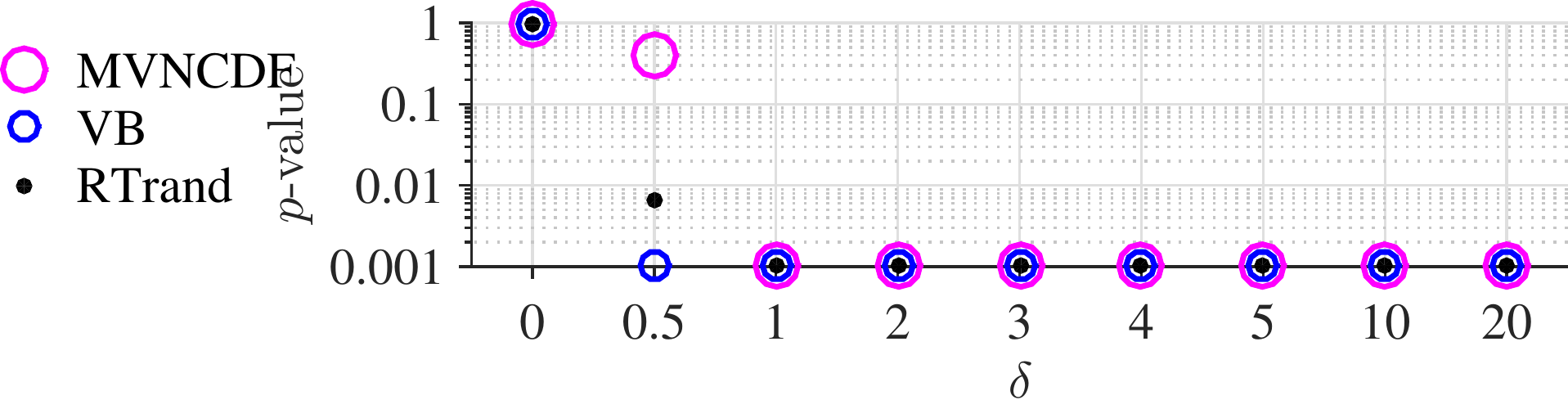}
\hspace{-0.5ex}
\includegraphics[width=0.437 \columnwidth]{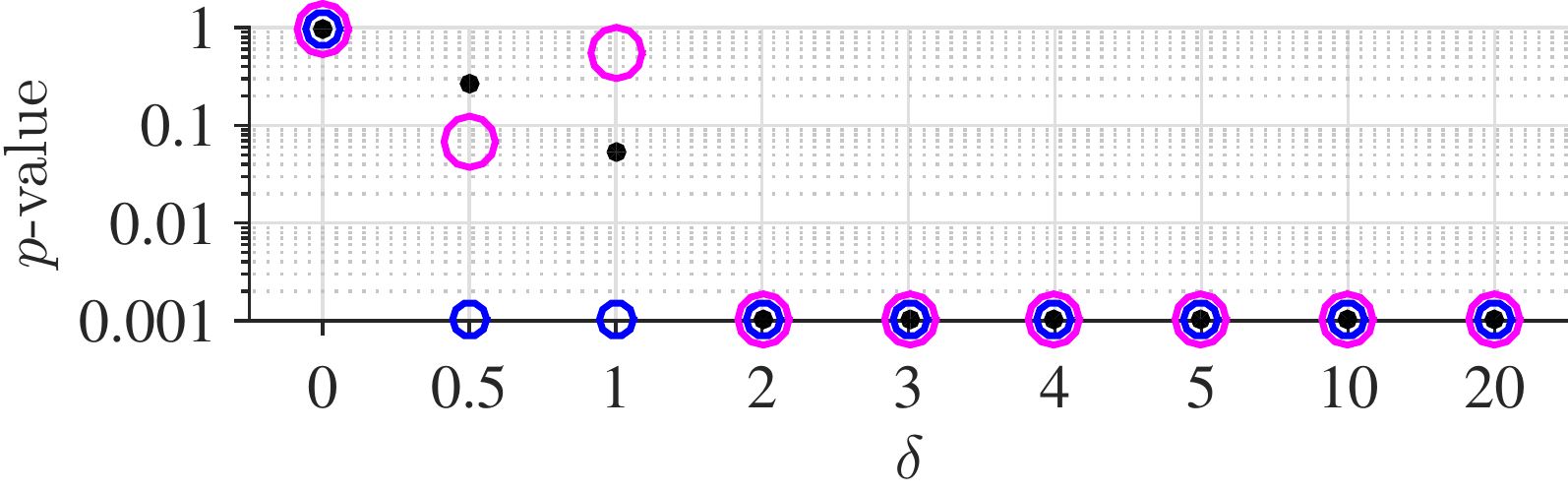}
\vspace{-3mm}
\caption{With large $\delta$ values RTopt is closer to PF than RTrand but less accurate than computationally heavy MVNCDF (upper row). $p$-values of two-sided Wilcoxon signed rank test (bottom row) show that the differences from RTopt are significant with large $\delta$. (left) $\rho\!=\!1^2$, (right) $\rho\!=\!20^2$.} \label{fig:mvncdftest}
\vspace{3ex}
\includegraphics[width=0.54 \columnwidth]{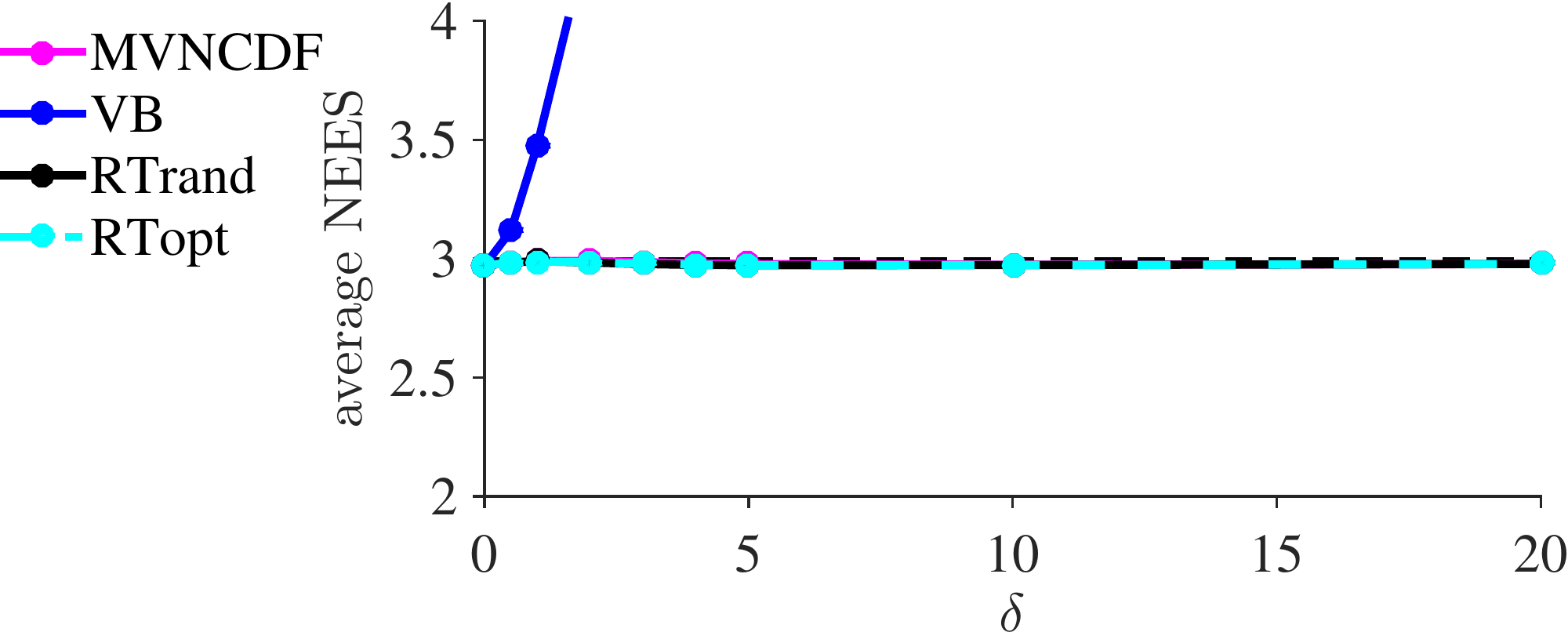}
\hfil
\includegraphics[width=0.43 \columnwidth]{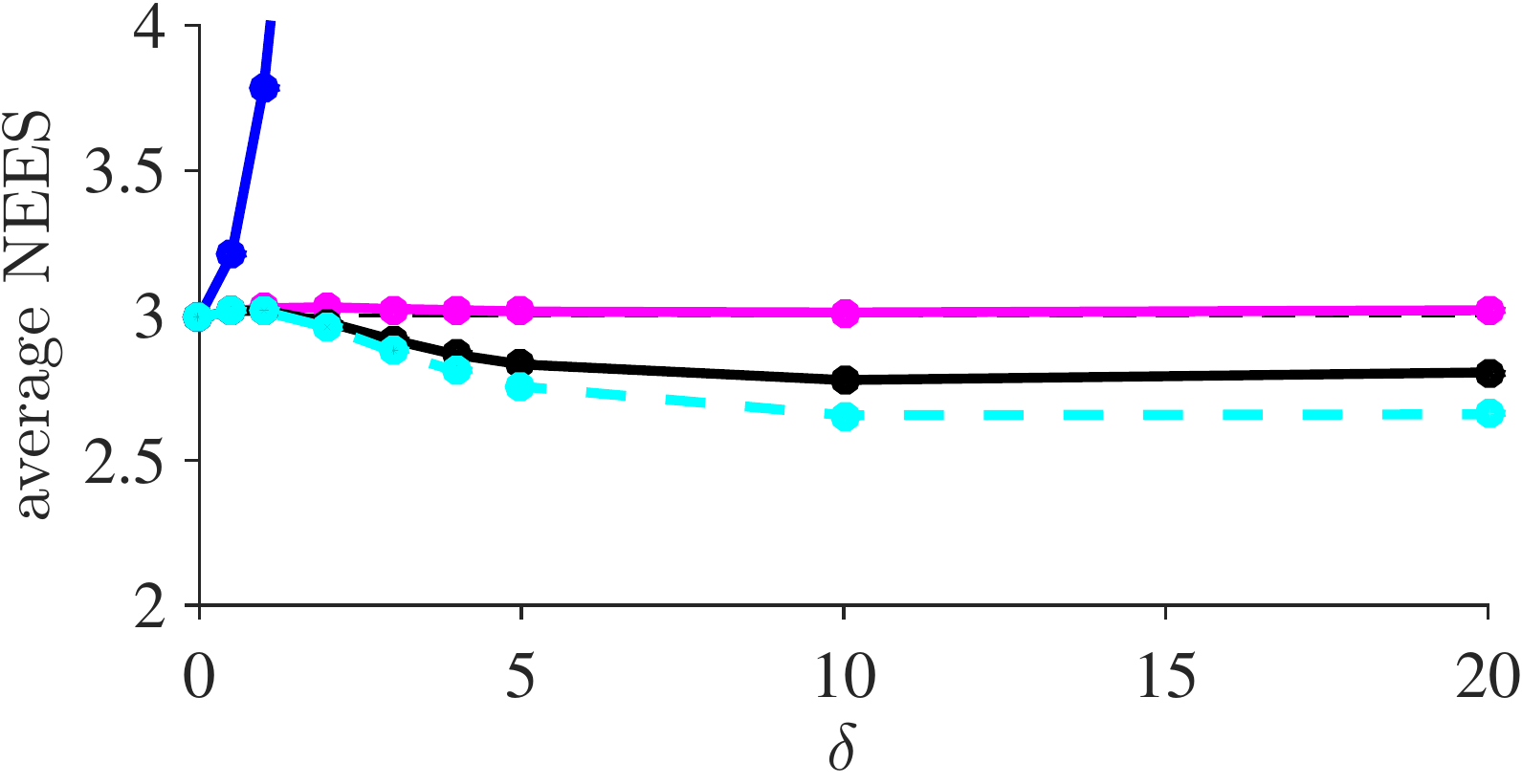}
\vspace{-2mm}
\caption{RTopt's NEES is closest to the optimal value 3, so recursive truncation gives the most realistic covariance matrix. (left) $\rho\!=\!1^2$, (right) $\rho\!=\!20^2$.} \label{fig:neestest}
\end{figure}

\subsection{Skew-$t$ inference}

In this section, the proposed skew $t$ filter (\STRTVBF) is compared with state-of-the-art filters using numerical simulations of a 100-step trajectory. The compared methods are a bootstrap-type PF, STVBF \cite{nurminen2015a}, $t$ variational Bayes filter (TVBF) \cite{piche2012}, and Kalman filter (KF) with measurement validation gating \cite[Ch.\ 5.7.2]{bar-shalom} that discards the measurement components whose normalized innovation squared is larger than the $\chi_1^2$-distribution's 99\,\% quantile. TVBF and KF's parameters are numerically optimized maximum expected likelihood parameters. The results are based on 1000 Monte Carlo replications.

Fig.\ \ref{fig:time_vs_acc} illustrates the filter iterations' convergence. The figure shows that the proposed \STRTVBF converges within 5 VB iterations and outperforms the other filters except for PF already with 2 VB iterations. Furthermore, Fig.\ \ref{fig:time_vs_acc} shows that \STRTVBF's converged state is close to the PF's converged state in RMSE, and PF can require as many as 10\,000 particles to outperform \STRTVBF. \STRTVBF also converges faster than STVBF when the process variance parameter $q$ is large. With a small $q$, STVBF with a small number of VB iterations can give a lower RMSE than the converged STVBF. The reason for this is probably that in the first iterations STVBF accommodates outliers by decreasing the $\Lambda_k$ estimates, which also affects the covariance, while in the later iterations $u_k$ estimates are increased, which makes the mean more accurate but underestimates the covariance.
\begin{figure}[t]
\centering
\includegraphics[width=0.48\columnwidth]{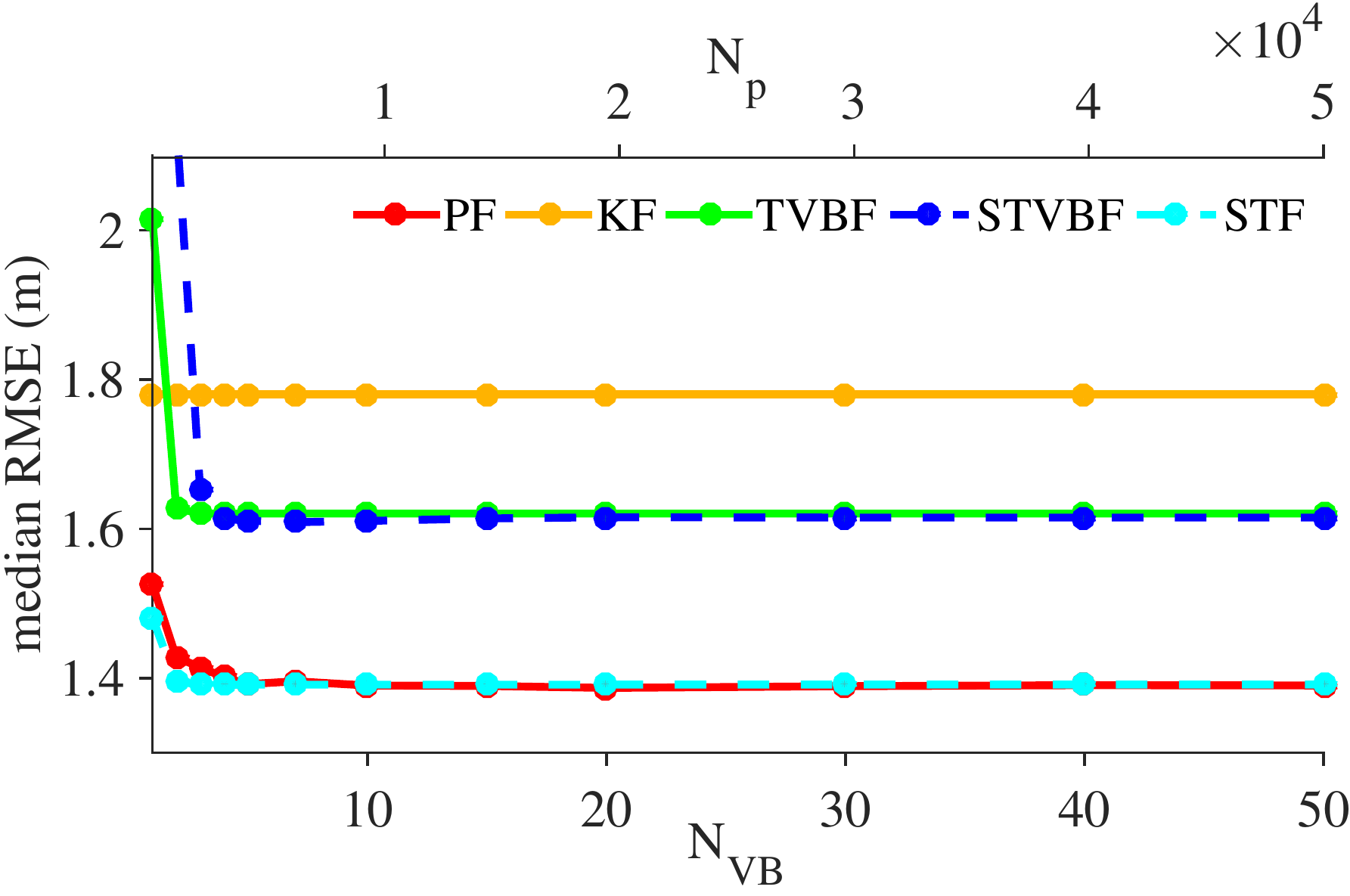}
\hfil
\includegraphics[width=0.48\columnwidth]{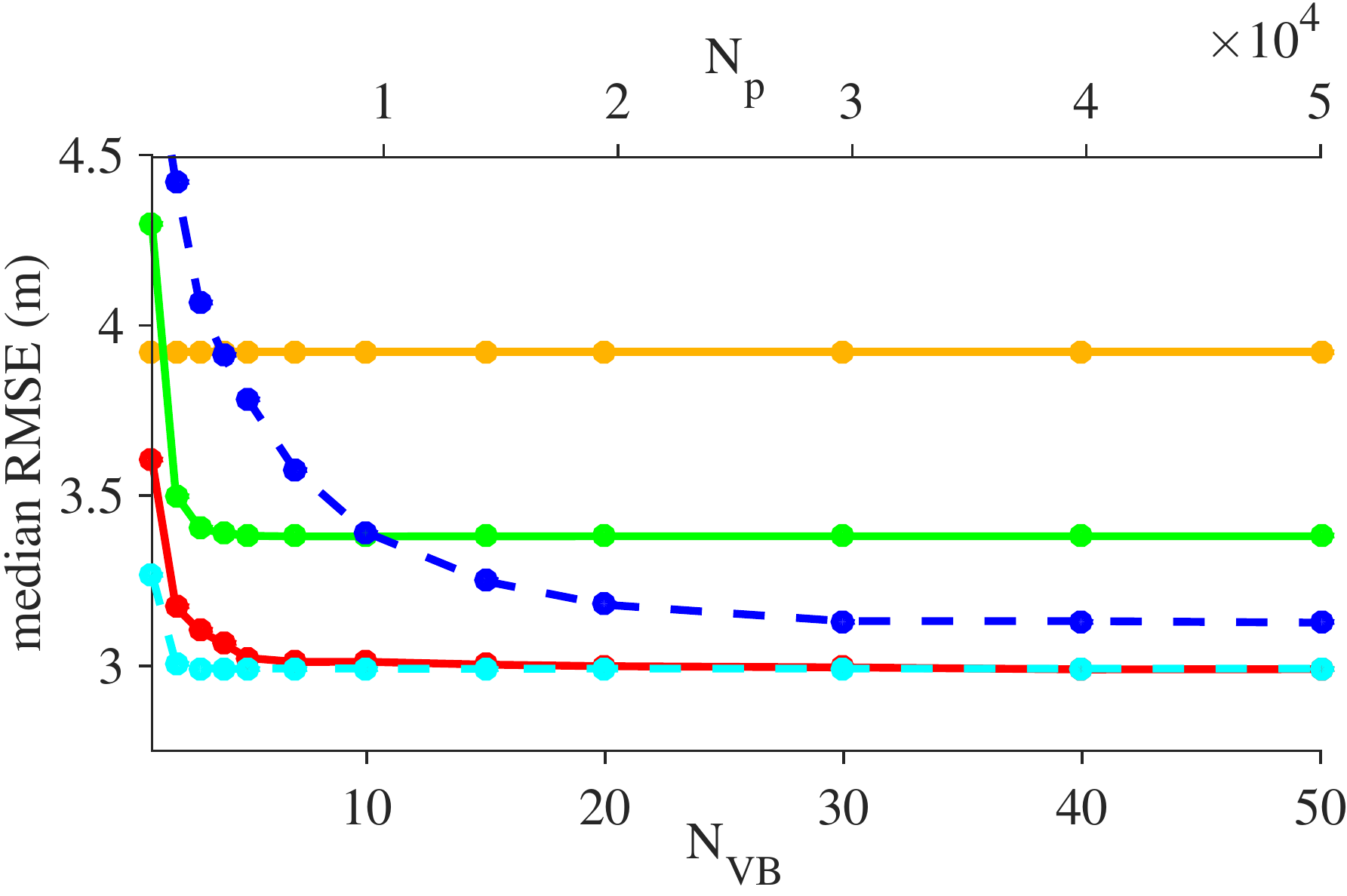}
\vspace{-2mm}
\caption{\STRTVBF converges in five iterations. The required number of PF particles can be 10.000. (left) $q\!=\!0.5,\ \delta\!=\!5$, (right) $q\!=\!5,\ \delta\!=\!5$.}
\label{fig:time_vs_acc}
\vspace{2ex}
\centering
\includegraphics[clip,width=0.48\columnwidth]{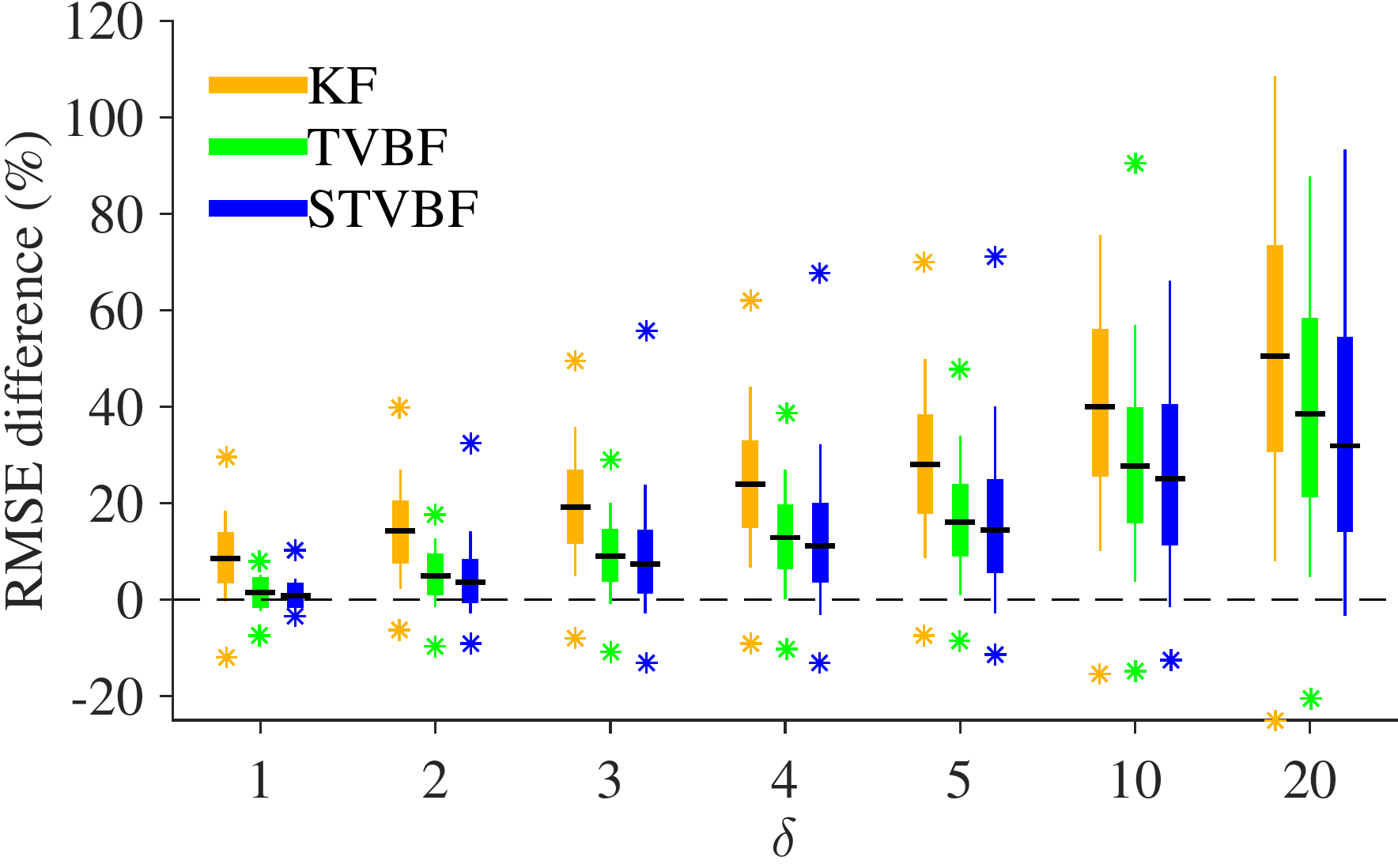}
\hfil
\includegraphics[width=0.48\columnwidth]{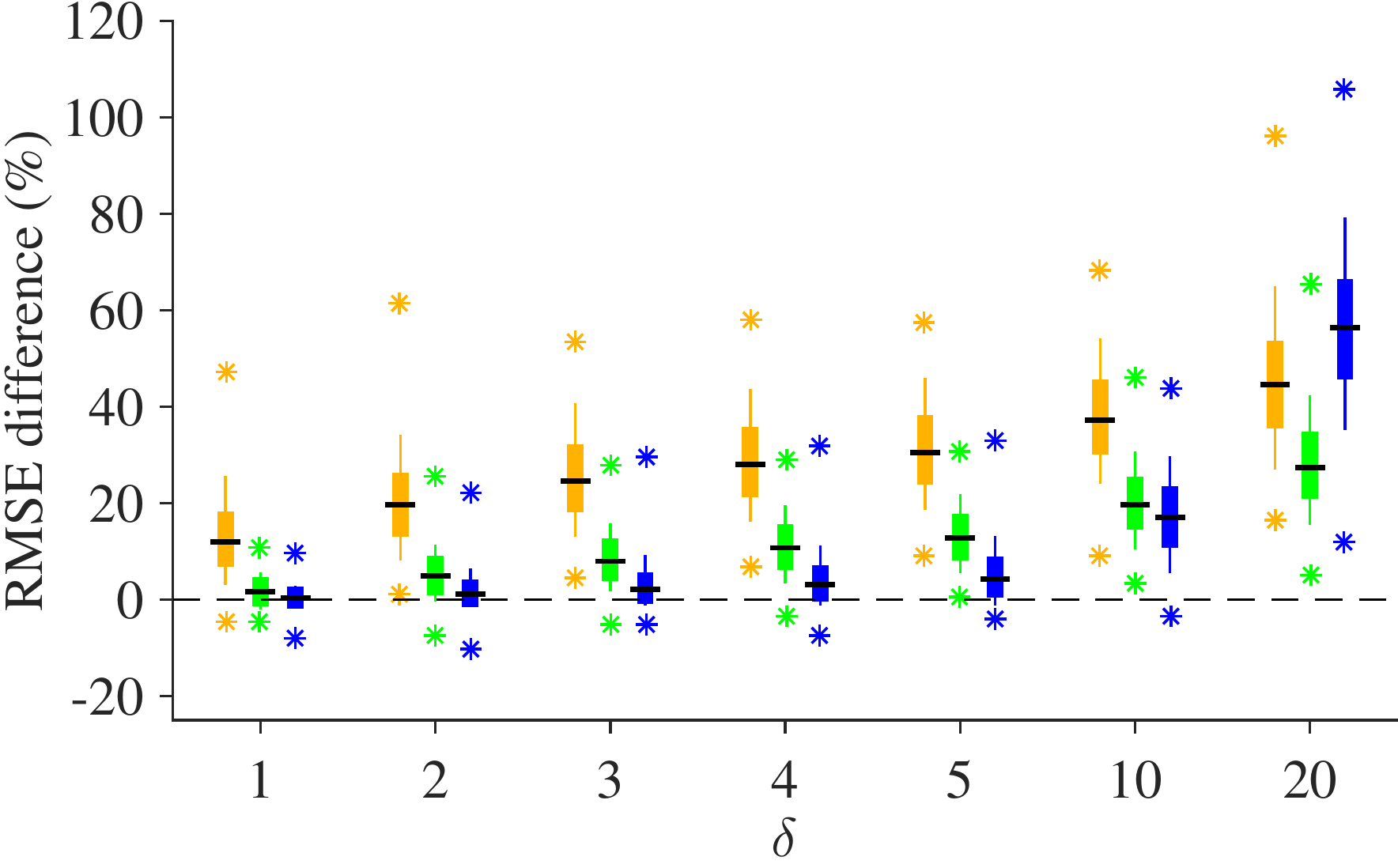}
\vspace{-2mm}
\caption{\STRTVBF outperforms the comparison methods with skew-$t$-distributed noise. RMSE differences per cent of the \STRTVBF's RMSE. The relative differences increase as $\delta$ is increased. (left) $q=0.5$, (right) $q=5$.} \label{fig:rmse_skewt}
\vspace{2ex}
\centering
\includegraphics[clip,width=0.75\columnwidth]{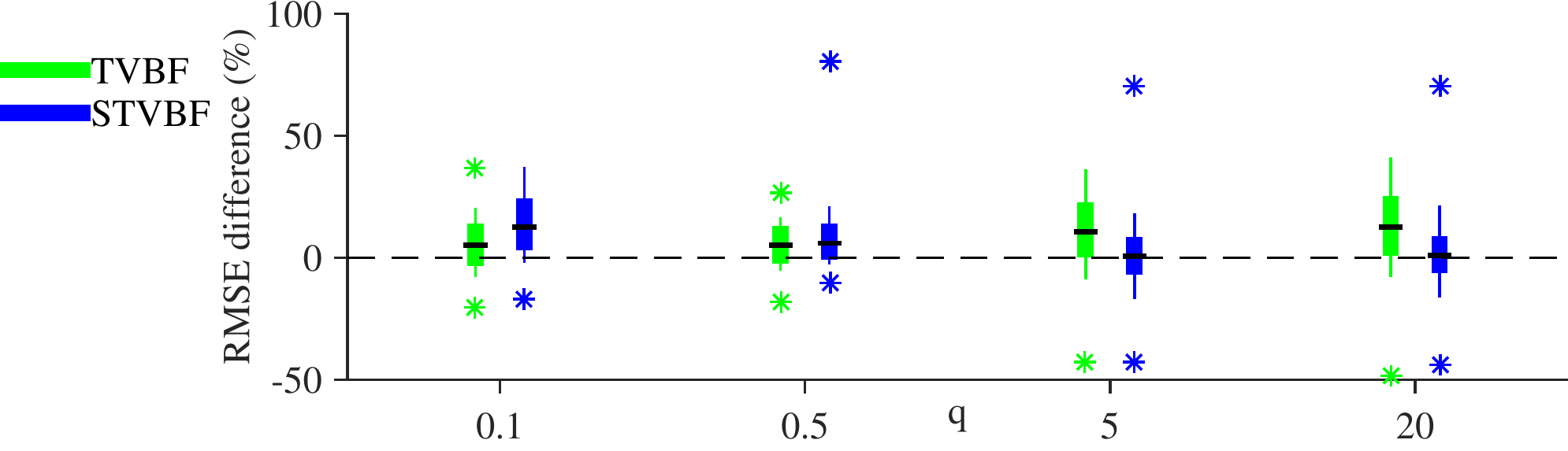}
\vspace{-2mm}
\caption{\STRTVBF outperforms TVBF and STVBF with UWB noise. RMSE differences per cent of the \STRTVBF's RMSE.} \label{fig:uwbtest}
\vspace{2ex}
\centering
\includegraphics[width=0.48\columnwidth]{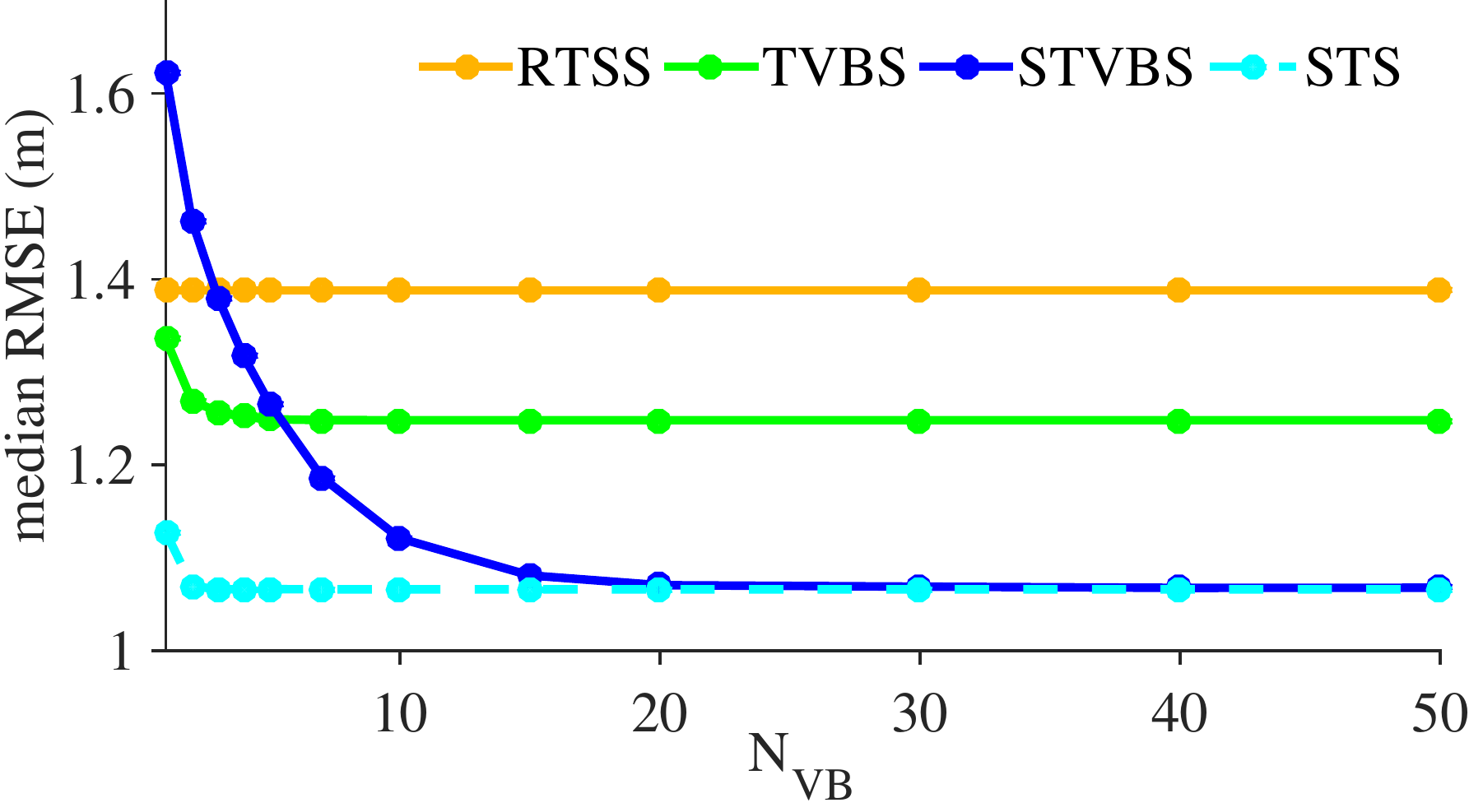}
\hfil
\includegraphics[width=0.48\columnwidth]{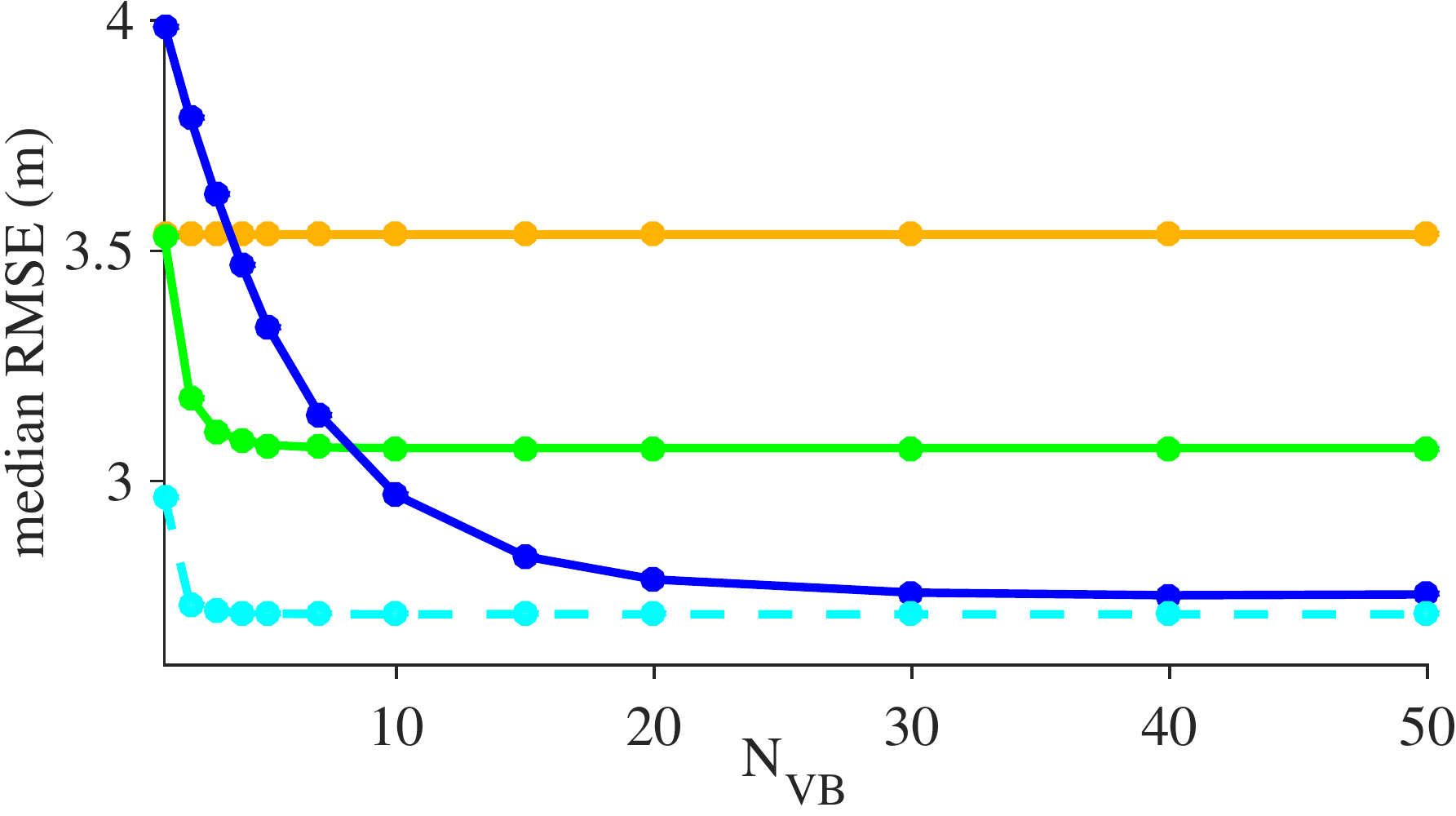}
\vspace{-2mm}
\caption{Five \STRTVBS  iterations give the converged state's RMSE. (left) $q\!=\!0.5,\ \delta\!=\!5$, (right) $q\!=\!5,\ \delta\!=\!5$.}
\label{fig:time_vs_acc_smoother}
\end{figure}

Fig.\ \ref{fig:rmse_skewt} shows the distributions of the RMSE differences from the \STRTVBF's RMSE as percentages of the \STRTVBF's RMSE. STF clearly has the smallest RMSE when $\delta\!\geq\!3$. Unlike STVBF, the new \STRTVBF improves accuracy even with small $q$, which can be explained by the improved covariance approximation.

Fig.\ \ref{fig:uwbtest} shows the results of a test where the measurement noise in \eqref{eq:measmodel} is generated from the histogram distribution of the UWB time-of-flight data set used in \cite{nurminen2015b}. The filters use the maximum likelihood parameters fitted to the data set numerically with the degrees-of-freedom parameters fixed to 4. The proposed method \STRTVBF has the lowest RMSE also in this test, which shows that the method is robust to deviations from the assumed distribution and thus usable with real data.

The proposed smoother is also tested with measurements generated from \eqref{eq:measmodel}. The compared smoothers are the proposed skew $t$ smoother (\STRTVBS), skew $t$ variational Bayes Smoother (STVBS) \cite{nurminen2015a}, $t$ variational Bayes smoother (TVBS) \cite{piche2012}, and the RTSS with 99\,\% measurement validation gating \cite{RTS-1965}. Fig.\ \ref{fig:time_vs_acc_smoother} shows that \STRTVBS has lower RMSE than the smoothers based on symmetric distributions. Furthermore, \STRTVBF's VB iteration converges in five iterations, so it is faster than STVBF.

\section{Conclusions} \label{sec:conclusions}
We have proposed a novel approximate filter and smoother for linear state-space models with heavy-tailed and skewed measurement noise distribution. The algorithms are based on the variational Bayes approximation, where some posterior independence approximations are removed from the earlier versions of the algorithms to avoid significant underestimation of the posterior covariance matrix. Removal of independence approximations is enabled by the recursive truncation algorithm for approximating the mean and covariance matrix of truncated multivariate normal distribution. An optimal processing sequence is given for the recursive truncation.
\vfill\clearpage\onecolumn
\appendices
\section{Derivations for the smoother} \label{sec:smoother}
We derive the expectations for the iterations of the variational Bayes smoother approximating the joint smoothing density
\begin{align}
p&(x_{1:K},u_{1:K},\Lambda_{1:K}|y_{1:K}) \propto p(x_{1:K},u_{1:K},\Lambda_{1:K},y_{1:K})\\
&= p(x_1)\prod_{l=1}^{K-1} p(x_{l+1}|x_l)\prod_{k=1}^K p(y_k|x_k,u_k,\Lambda_k)\,p(u_k|\Lambda_k)\,p(\Lambda_k) \\
&=\N(x_1;x_{1|0},P_{1|0})\prod_{l=1}^{K-1}\N(x_{l+1};Ax_l,Q) \cdot \prod_{k=1}^K\left\{ \N(y_k;Cx_k+\Delta u_k,\Lambda_k^{-1}R)\,\N_+(u_k; 0,\Lambda_k^{-1}) \prod_{i=1}^{n_y}\G\left([\Lambda_k]_{ii};\frac{\nu_i}{2},\frac{\nu_i}{2}\right) \right\}
\end{align}
which is approximated by a factorized probability density function (PDF) in the form
\begin{align}
\label{eq:factors}
p(x_{1:K},&u_{1:K},\Lambda_{1:K}|y_{1:K})\approx q_{xu}(x_{1:K},u_{1:K})\,q_\Lambda(\Lambda_{1:K}).
\end{align}
The VB solutions $\hat{q}_{xu}$ and $\hat{q}_\Lambda$ can be obtained by cyclic iteration of
\begin{subequations}
\label{eq:IterativeOptimization}
\begin{align}
\log {q}_{xu}(x_{1:K},u_{1:K}) \leftarrow& \E_{{q}_{\Lambda}}[\log p(y_{1:K},x_{1:K},u_{1:K},\Lambda_{1:K})]+c_{xu}\label{eq:IterativeOptimizationxu}\\
\log {q}_{\Lambda}(\Lambda_{1:K}) \leftarrow& \E_{{q}_{xu}}[\log p(y_{1:K},x_{1:K},u_{1:K},\Lambda_{1:K})]+c_{\Lambda}\label{eq:IterativeOptimizationL}
\end{align}
\end{subequations}
where the expected values are taken with respect to the current $q_{xu}$ and $q_\Lambda$, and $c_{xu}$ and $c_\Lambda$  are constants with respect to the variables $\left[\begin{smallmatrix}x_k\\u_k\end{smallmatrix}\right]$ and $\Lambda_k$,  respectively ~\cite[Chapter 10]{Bishop2007}\cite{TzikasLG2008}. This appendix gives the derivations for one iteration of \eqref{eq:IterativeOptimization}. For brevity all constant values are denoted by $c$. The logarithm of the joint smoothing distribution is
\begin{equation}
\begin{split}
\log p(&x_{1:K}, u_{1:K}, \Lambda_{1:K}, y_{1:K}) = \log\N(x_1; x_{1|0},P_{1|0}) + \sum_{l=1}^K \log\N(x_{l+1}; A x_l, Q) \\
&+ \sum_{k=1}^{K-1} \left\{ \log \N(y_k; C x_k+\Delta u_k, \Lambda_k^{-1}R) + \log\N_+(u_k; 0,\Lambda_k^{-1}) \right\} + \sum_{k=1}^K \sum_{i=1}^{n_y} \log\G([\Lambda_k]_{ii}; \tfrac{\nu_i}{2}, \tfrac{\nu_i}{2}),
\end{split}
\end{equation}

\subsection{Derivations for $q_{xu}$} \label{sec:Smootherq_xu}
Eq.\ \eqref{eq:IterativeOptimizationxu} gives
\begin{align}
\log& q_{xu}(x_{1:K},u_{1:K}) = \log\N(x_1; x_{1|0}, P_{1|0}) + \sum_{l=1}^{K-1} \log\N(x_{l+1}; Ax_l, Q) \nonumber\\
&+ \sum_{k=1}^K \E_{q_\Lambda} [\log\N(y_k; Cx_k + \Delta u_k, \Lambda_k^{-1} R)+\log\N_+(u_k;0,\Lambda_k^{-1})] + c \\
=& \log\N(x_1; x_{1|0}, P_{1|0}) + \sum_{l=1}^{K-1} \log\N(x_{l+1}; Ax_l, Q) \nonumber\\
&- \frac{1}{2} \sum_{k=1}^K \E_{q_\Lambda} [(y_k-Cx_k-\Delta u_k)^\t R^{-1} \Lambda_k (y_k-Cx_k-\Delta u_k) + u_k^\t \Lambda_k u_k] + c \\
=& \log\N(x_1; x_{1|0}, P_{1|0}) + \sum_{l=1}^{K-1} \log\N(x_{l+1}; Ax_l, Q) \nonumber\\
&- \frac{1}{2} \sum_{k=1}^K \left\{ (y_k-Cx_k-\Delta u_k)^\t R^{-1} \Lambda_{k|K} (y_k-Cx_k-\Delta u_k) + u_k^\t \Lambda_{k|K} u_k \right\} + c \\
=&  \log\N(x_1; x_{1|0}, P_{1|0}) + \sum_{l=1}^{K-1} \log\N(x_{l+1}; Ax_l, Q) \nonumber\\
&+ \sum_{k=1}^K \left\{ \log\N(y_k; Ax_k + \Delta u_k, \Lambda_{k|K}^{-1} R) + \log\N(u_k;0,\Lambda_{k|K}^{-1}) \right\} + c \\
=&  \log\N\left( \begin{bmatrix} x_1 \\ u_1 \end{bmatrix} ; \begin{bmatrix} x_{1|0}\\0 \end{bmatrix} , \begin{bmatrix} P_{1|0} & \zeros \\\zeros & \Lambda_{1|K}^{-1} \end{bmatrix} \right) + \sum_{l=1}^{K-1} \log \N\left( \begin{bmatrix} x_{l+1}\\u_{l+1} \end{bmatrix}; \begin{bmatrix} A & \zeros \\\zeros & \zeros \end{bmatrix} \begin{bmatrix} x_l\\u_l \end{bmatrix}, \begin{bmatrix} Q & \zeros \\\zeros & \Lambda_{l+1|K}^{-1} \end{bmatrix} \right) \nonumber\\
&+ \log\N\left( y_k; \begin{bmatrix} C & \Delta \end{bmatrix} \begin{bmatrix} x_k\\u_k \end{bmatrix}, \Lambda_{k|K}^{-1} R \right) + c,\ u_{1:K}\geq0 ,
\end{align}
where $\Lambda_{k|K} \triangleq \E_{q_\Lambda}[\Lambda_k]$ is derived in Section \ref{sec:Smootherq_L}, and $u_{1:K}\geq0$ means that all the components of all $u_k$ are required to be nonnegative for each $k=1\cdots K$. Up to the truncation of the $u$ components, $q_{xu}(x_{1:K}, u_{1:K})$ has thus the same form as the joint smoothing posterior of a linear state-space model with the state transition matrix $\widetilde{A} \triangleq \left[ \begin{smallmatrix} A&\zeros\\\zeros&\zeros \end{smallmatrix} \right]$, process noise covariance matrix $\widetilde{Q_k} \triangleq \left[ \begin{smallmatrix} Q&\zeros\\\zeros&\Lambda_{k+1|K}^{-1} \end{smallmatrix} \right]$, measurement model matrix $\widetilde{C} \triangleq \left[ \begin{smallmatrix} C & \Delta  \end{smallmatrix} \right]$, and measurement noise covariance matrix $\widetilde{R} \triangleq \Lambda_{k|K}^{-1} R$. Let us denote the PDFs related to this state-space model with $\widetilde{p}$.

It would be possible to compute the truncated multivariate normal posterior of the joint smoothing distribution $\widetilde{p}\left(\left[\begin{smallmatrix}x_{1:K}\\u_{1:K}\end{smallmatrix}\right]|y_{1:K}\right)$, 
and account for the truncation of $u_{1:K}$ to the positive orthant using the recursive truncation. However, this would be impractical with large $K$ due to the large dimensionality $K\times (n_x+n_y)$. A feasible solution is to approximate each filtering distribution in the Rauch--Tung--Striebel smoother's (RTSS \cite{RTS-1965}) forward filtering step with a multivariate normal distribution by
\begin{align}
\widetilde{p}( x_k,u_k|y_{1:k}) &= \frac{1}{C}\,\N\left(\begin{bmatrix} x_k\\u_k \end{bmatrix}; z_{k|k}', Z_{k|k}' \right) \cdot [u_k\geq0]\\
&\approx \N\left(\begin{bmatrix} x_k\\u_k \end{bmatrix}; z_{k|k}, Z_{k|k} \right) 
\end{align}
for each $k=1\cdots K$, where $[u_k\geq0]$ is the Iverson bracket notation
\[
[u_k\geq0] = \left\{ \begin{array}{ll} 1, &\text{if all components of } u_k \text{ are non-negative}\\ 0, &\text{otherwise} \end{array} \right. ,
\]
$C$ is the normalization factor, and $z_{k|k} \triangleq \E_{\widetilde{p}}\left[\left[\begin{smallmatrix}x_k\\u_k\end{smallmatrix}\right]|y_{1:k}\right]$ and $Z_{k|k} \triangleq \var_{\widetilde{p}}\left[\left[\begin{smallmatrix}x_k\\u_k\end{smallmatrix}\right]|y_{1:k}\right]$ are approximated using the recursive truncation. Given the multivariate normal approximations of the filtering posteriors $\widetilde{p}(x_k,u_k|y_{1:k})$, by Lemma \ref{lem:gaussian_smoother} the backward recursion of the RTSS gives multivariate normal approximations of the smoothing posteriors $\widetilde{p}(x_k,u_k|y_{1:K})$. The quantities required in the derivations of Section \ref{sec:Smootherq_L} are the expectations of the smoother posteriors $x_{k|K} \triangleq \E_{q_{xu}}[x_k]$, $u_{k|K} \triangleq \E_{q_{xu}}[u_k]$, and the covariance matrices $Z_{k|K} \triangleq \var_{q_{xu}}\left[\begin{smallmatrix} x_k\\u_k \end{smallmatrix}\right]$ and $U_{k|K} \triangleq \var_{q_{xu}}[u_k]$.

\begin{lemma} \label{lem:gaussian_smoother}
\newcommand{\xu}{z}
\newcommand{\cxu}{Z}
Let $\{\xu_k\}_{k=1}^K$ be a linear--Gaussian process, and $\{y_k\}_{k=1}^K$ a measurement process such that
\begin{align}
\xu_1 &\sim \N(\xu_{1|0},\cxu_{1|0}) \\
\xu_k | \xu_{k-1} &\sim \N(A \xu_{k-1}, Q) \\
y_k | \xu_k &\sim \text{(a known distribution)} .
\end{align}
with the standard Markovianity assumptions. Then, if the filtering posterior $p(\xu_k|y_{1:k})$ is a multivariate normal distribution for each $k$, then for each $k < K$
\begin{equation} \label{eq:smootherlemma_assume}
\xu_k | y_{1:K} \sim \N(\xu_{k|K}, \cxu_{k|K}), 
\end{equation}
where
\begin{align}
\xu_{k|K} &= \xu_{k|k} + G_k(\xu_{k+1|K}-A \xu_{k|k}) ,\\
\cxu_{k|K} &= \cxu_{k|k} + G_k(\cxu_{k+1|K}-A\cxu_{k|k}A^\t-Q)G_k^\t ,\\
G_k &= \cxu_{k|k} A^\t (A\cxu_{k|k}A^\t + Q)^{-1} ,
\end{align}
and $\xu_{k|k}$ and $\cxu_{k|k}$ are the mean and covariance matrix of the filtering posterior $p(\xu_k|y_{1:k})$.
\end{lemma}
\begin{proof}
\newcommand{\xu}{z}
\newcommand{\cxu}{Z}
The proof is mostly similar to the proof of \cite[Theorem 8.2]{sarkka2013}. First, assume that
\begin{equation} \label{eq:induction_assumption}
\xu_{k+1}|y_{1:K} \sim \N(\xu_{k+1|K}, \cxu_{k+1|K}) .
\end{equation}
for some $k<K$. The joint conditional distribution of $\xu_k$ and $\xu_{k+1}$ is then
\begin{align}
p(\xu_k,\xu_{k+1} |y_{1:k})
&= p(\xu_{k+1} | \xu_k)\,p(\xu_k | y_{1:k})\ \ ||\,\text{Markovianity assumption} \\
&= \N(\xu_{k+1}; A\xu_k, Q)\,\N(\xu_k; \xu_{k|k}, \cxu_{k|k}) \\
&= \N\left( \begin{bmatrix} \xu_k\\\xu_{k+1} \end{bmatrix}; \begin{bmatrix} \xu_{k|k}\\A\xu_{k|k} \end{bmatrix}, \begin{bmatrix} \cxu_{k|k} & \cxu_{k|k}A^\t\\A\cxu_{k|k} & A\cxu_{k|k}A^\t+Q \end{bmatrix} \right) ,
\end{align}
so by the conditioning rule of the multivariate normal distribution
\begin{align}
p(\xu_k|\xu_{k+1},y_{1:k})
&= \N(\xu_k; \xu_{k|k}+G_k(\xu_{k+1}-A\xu_{k|k}), \cxu_{k|k}-Z_{k|k}A^\t (AZ_{k|k}A^\t+Q)^{-1} A\cxu_{k|k}) \\
&= \N(\xu_k; \xu_{k|k}+G_k(\xu_{k+1}-A\xu_{k|k}), \cxu_{k|k}-G_k (A\cxu_{k|k}A^\t + Q) G_k^\t) .
\end{align}
We use this formula in
\begin{align}
p(\xu_k,\xu_{k+1} | y_{1:K})
&= p(\xu_k | \xu_{k+1}, y_{1:K})\,p(\xu_{k+1}|y_{1:K}) \\
&= p(\xu_k | \xu_{k+1}, y_{1:k})\,p(\xu_{k+1}|y_{1:K})\ \ ||\,\text{Markovianity assumption} \\
&= \N(\xu_k; \xu_{k|k}+G_k(\xu_{k+1}-A\xu_{k|k}), \cxu_{k|k}-G_k (A\cxu_{k|k}A^\t + Q) G_k^\t)\,\N(\xu_{k+1}|\xu_{k+1|K}, \cxu_{k+1|K}) \\
&= \N\left( \begin{bmatrix} \xu_k\\\xu_{k+1} \end{bmatrix}; \begin{bmatrix} \xu_{k|k}+G_k(\xu_{k+1|K}-A\xu_{k|k})\\\bullet \end{bmatrix}, \begin{bmatrix} G_k\cxu_{k+1|K}G_k^\t + \cxu_{k|k}-G_k (A\cxu_{k|k}A^\t + Q) G_k^\t & \bullet\\\bullet&\bullet \end{bmatrix} \right) \\
&= \N\left( \begin{bmatrix} \xu_k\\\xu_{k+1} \end{bmatrix}; \begin{bmatrix} \xu_{k|k}+G_k(\xu_{k+1|K}-A\xu_{k|k})\\\bullet \end{bmatrix}, \begin{bmatrix} \cxu_{k|k}+G_k (\cxu_{k+1|K}-A\cxu_{k|k}A^\t-Q) G_k^\t & \bullet\\\bullet&\bullet \end{bmatrix} \right),
\end{align}
so
\begin{align} 
p(\xu_k|y_{1:K}) &= \N(\xu_k;  \xu_{k|k}+G_k(\xu_{k+1|K}-A\xu_{k|k}), \cxu_{k|k}+G_k (\cxu_{k+1|K}-A\cxu_{k|k}A^\t-Q) G_k^\t) \\
&= \N(\xu_k; \xu_{k|K}, \cxu_{k|K}) . \label{eq:induction_proof}
\end{align}
Because $\xu_K|y_{1:K} \sim \N(\xu_{K|K}, Z_{K|K})$, and because \eqref{eq:induction_assumption} implies \eqref{eq:induction_proof}, the statement holds by the induction argument.
\end{proof}

\subsection{Derivations for $q_\Lambda$} \label{sec:Smootherq_L}

Eq.\ \eqref{eq:IterativeOptimizationL} gives
\begin{align}
\log q_\Lambda(\Lambda_{1:K})=&\sum_{k=1}^K \left\{\E_{q_{xu}}\left[\log \N(y_k;Cx_k+\Delta u_k,\Lambda_k^{-1}R) +\log \N_+(u_k; 0,\Lambda_k^{-1})\right]\right\} +\sum_{k=1}^K \sum_{i=1}^{n_y}\log\G\left([\Lambda_k]_{ii};\frac{\nu_i}{2},\frac{\nu_i}{2}\right)+c.\end{align}
Therefore, $q_\Lambda(\Lambda_{1:K})=\prod_{k=1}^K q_\Lambda(\Lambda_{k})$ where
\begin{align}
\log& q_\Lambda(\Lambda_k)= -\frac{1}{2}\E_{{q}_{xu}}[\tr\{(y_k-Cx_k-\Delta u_k)(y_k-Cx_k-\Delta u_k)^\t R^{-1}\Lambda_k\}]\nonumber\\
& - \frac{1}{2}\E_{{q}_{xu}}[\tr\{ u_k  u_k^\t{}\Lambda_k\}] + \sum_{i=1}^{n_y}\left( \frac{\nu_i}{2}\log[\Lambda_k]_{ii}- \frac{\nu_i}{2} [\Lambda_k]_{ii} \right)+c\\
=&-\frac{1}{2} \tr\left\{\left((y_k-Cx_{k|K}-\Delta u_{k|K})(y_k-Cx_{k|K}-\Delta u_{k|K})^\t+\begin{bmatrix}C&\Delta\end{bmatrix} Z_{k|K} \begin{bmatrix}C^\t\\\Delta^\t\end{bmatrix}\right) R^{-1} \Lambda\right\} \\
&-\frac{1}{2}\tr\left\{ (u_{k|K} u_{k|K}^\t+U_{k|K}) \Lambda_k \right\} + \sum_{i=1}^{n_y}\left( \frac{\nu_i}{2}\log[\Lambda_k]_{ii}- \frac{\nu_i}{2} [\Lambda_k]_{ii} \right)+c \\
=& \sum_{i=1}^{n_y} \left( \frac{\nu_i}{2} \log[\Lambda_k]_{ii} - \frac{\nu_i+[\Psi_k]_{ii}}{2}[\Lambda_k]_{ii} \right) +c 
\end{align}
where
\begin{equation}
\begin{split}
\Psi_k &= (y_k-Cx_{k|K}-\Delta u_{k|K})(y_k-Cx_{k|K}-\Delta u_{k|K})^\t R^{-1} + \begin{bmatrix}C&\Delta\end{bmatrix} Z_{k|K} \begin{bmatrix}C^\t\\\Delta^\t\end{bmatrix} R^{-1} +u_{k|K}u_{k|K}^\t+U_{k|K}. \label{eq:Filterpsi}
\end{split}
\end{equation}
Therefore,
\begin{align}
q_\Lambda&(\Lambda_k)=\prod_{i=1}^{n_y}\G\left([\Lambda_k]_{ii};\frac{\nu_i}{2}+1,\frac{\nu_i+[\Psi_k]_{ii}}{2}\right).
\end{align}
Note that only the diagonal elements of the matrix $\Psi_k$ are required. In the derivations of Section \ref{sec:Smootherq_xu}, $\Lambda_{k|K} \triangleq \E_{q_\Lambda}[\Lambda_k]$ is required. $\E_{q_\Lambda}[\Lambda_k]$ is a diagonal matrix with the diagonal elements
\begin{align}
[\Lambda_{k|K}]_{ii}=\frac{\nu_i+2}{\nu_i+[\Psi_k]_{ii}}.
\end{align}

\vfill
\section{Derivations for the Filter} \label{sec:filter}

Suppose that at time index $k$ the measurement vector $y_k$ is available, and the prediction PDF $p(x_k|y_{1:k-1})$ is
\begin{align}
p(x_k|y_{1:k-1}) = \N(x_k;x_{k|k-1},P_{k|k-1}).
\end{align}
Then, using Bayes' theorem  the joint filtering posterior PDF is
\begin{align}
p(x_k,u_k,\Lambda_k|y_{1:k}) &\propto  p(y_k,x_k,u_k,\Lambda_k|y_{1:k-1})\\
&= p(y_k|x_k,u_k,\Lambda_k)\,p(x_k|y_{1:k-1})\,p(u_k|\Lambda_k)\,p(\Lambda_k)\\
&= \N(y_k;Cx_k+\Delta u_k,\Lambda_k^{-1}R)\, \N(x_k;x_{k|k-1},P_{k|k-1})\, \N_+(u_k; 0,\Lambda_k^{-1})\, \prod_{i=1}^{n_y}\G\left([\Lambda_k]_{ii};\frac{\nu_i}{2},\frac{\nu_i}{2}\right).
\end{align}
This posterior is not analytically tractable. We seek an approximation in the form
\begin{align}
\label{eq:Filterfactors}
p(x_k,&u_k,\Lambda_k|y_{1:k})\approx q_{xu}(x_k,u_k)\,q_\Lambda(\Lambda_k) .
\end{align}
The VB solutions $\hat{q}_{xu}$ and $\hat{q}_\Lambda$ can be obtained by cyclic iteration of
\begin{subequations}
\label{eq:FilterIterativeOptimization}
\begin{align}
\log {q}_{xu}(x_k,u_k) \leftarrow& \E_{{q}_{\Lambda}}[\log p(y_k,x_k,u_k,\Lambda_k|y_{1:k-1})]+c_{xu}\label{eq:FilterIterativeOptimizationxu}\\
\log {q}_{\Lambda}(\Lambda_k) \leftarrow& \E_{{q}_{xu}}[\log p(y_k,x_k,u_k,\Lambda_k|y_{1:k-1})]+c_{\Lambda}\label{eq:FilterIterativeOptimizationL}
\end{align}
\end{subequations}
where the expected values on the right hand sides of~\eqref{eq:FilterIterativeOptimization} are taken with respect to the current $q_{xu}$ and $q_\Lambda$, and $c_{xu}$ and $c_\Lambda$  are constants with respect to the variables $\left[\begin{smallmatrix}x_k\\u_k\end{smallmatrix}\right]$ and $\Lambda_k$,  respectively~\cite[Chapter 10]{Bishop2007}\cite{TzikasLG2008}. In sections \ref{sec:Filterq_xu} and \ref{sec:Filterq_L} the derivations for the variational solution \eqref{eq:FilterIterativeOptimization} are given. For brevity all constant values are denoted by $c$ in the derivations. The logarithm of the joint filtering posterior which is needed for the derivations is given by
\begin{align}
\log p(y_k,x_k,u_k,\Lambda_k|y_{1:k-1})= &-\frac{1}{2}(y_k-Cx_k-\Delta u_k)^\t R^{-1}\Lambda_k(y_k-Cx_k-\Delta u_k)\nonumber\\
&-\frac{1}{2}(x_k-x_{k|k-1})^\t{}P_{k|k-1}^{-1}(x_k-x_{k|k-1})\nonumber\\
& - \frac{1}{2}u_k^\t{} \Lambda_k u_k + \sum_{i=1}^{n_y} \left( \frac{\nu_i}{2}\log[\Lambda_k]_{ii}- \frac{\nu_i}{2} [\Lambda_k]_{ii} \right) +c,\ u_k\geq0 ,
\end{align}
where $u_k\geq0$ means that every component of $u_k$ is non-negative.

\subsection{ Derivations for $q_{xu}$}
\label{sec:Filterq_xu}
Using equation \eqref{eq:FilterIterativeOptimizationxu} we obtain
\begin{align}
\log q_{xu}(x_k,u_k)=&-\frac{1}{2}\E_{q_\Lambda}[ (y_k-Cx_k-\Delta u_k)^\t{} R^{-1} \Lambda_k (y_k-Cx_k-\Delta u_k)]\nonumber\\
&-\frac{1}{2}(x_k-x_{k|k-1})^\t{}P_{k|k-1}^{-1}(x_k-x_{k|k-1}) - \frac{1}{2}\E_{q_\Lambda}[ u_k^\t\Lambda_k u_k ] +c\\
=& -\frac{1}{2}\left(y_k-\begin{bmatrix} C & \Delta \end{bmatrix} \begin{bmatrix} x_k\\u_k \end{bmatrix}\right)^\t R^{-1} \Lambda_{k|k} \left(y_k-\begin{bmatrix} C & \Delta \end{bmatrix} \begin{bmatrix} x_k\\u_k \end{bmatrix} \right) \\
&-\frac{1}{2}\left( \begin{bmatrix} x_k\\u_k \end{bmatrix} - \begin{bmatrix} x_{k|k-1}\\0 \end{bmatrix} \right)^\t \begin{bmatrix} P_{k|k-1}&0\\0&\Lambda_{k|k}^{-1}\end{bmatrix}^{-1} \left( \begin{bmatrix} x_k\\u_k \end{bmatrix} - \begin{bmatrix} x_{k|k-1}\\0 \end{bmatrix} \right),\ u_k\geq0,
\end{align}
where $\Lambda_{k|k}\triangleq\E_{q_\Lambda}[\Lambda_k]$ is derived in section \ref{sec:Filterq_L}. Hence, 
\begin{align}
q_{xu}(x_k,u_k)\propto&\  \N\left(y_k; \begin{bmatrix}C&\Delta\end{bmatrix} \begin{bmatrix}x_k\\u_k\end{bmatrix},\Lambda_{k|k}^{-1}R\right)\N\left(\begin{bmatrix}x_k\\u_k\end{bmatrix};\begin{bmatrix}x_{k|k-1}\\0\end{bmatrix},\begin{bmatrix} P_{k|k-1}&0\\0&\Lambda_{k|k}^{-1}\end{bmatrix}\right) \cdot [u_k\geq0],
\end{align}
where $[u_k\geq0]$ is the Iverson bracket. By Kalman filter's \cite{Kalman60} measurement update, this becomes
\begin{align}
q_{xu}&(x_k,u_k)=\frac{1}{C}\,\N\left(\begin{bmatrix}x_k\\u_k\end{bmatrix};z_{k|k}',Z_{k|k}'\right) \cdot [u_k\geq0],
\end{align}
where
\begin{align}
z_{k|k}'&= \begin{bmatrix}x_{k|k-1}\\0\end{bmatrix}+K_k(y_k-Cx_{k|k-1}),  \label{eq:Filteru_kneeded}\\
Z_{k|k}'&= (I-K_k\begin{bmatrix}C&\Delta\end{bmatrix})\begin{bmatrix} P_{k|k-1}&0\\0&\Lambda_{k|k}^{-1}\end{bmatrix},\\
K_k&=\begin{bmatrix} P_{k|k-1}C^\t \\ \Lambda_{k|k}^{-1}\Delta^\t \end{bmatrix} (CP_{k|k-1}C^\t+\Delta\Lambda_{k|k}^{-1}\Delta^\t+\Lambda_{k|k}^{-1}R)^{-1}.
\end{align}

The first and second moments $x_{k|k} \triangleq \E_{q_{xu}}[x_k]$, $u_{k|k} \triangleq \E_{q_{xu}}[u_k]$, $Z_{k|k} \triangleq \Var_{q_{xu}}\left[\begin{smallmatrix}x_k\\u_k\end{smallmatrix}\right]$, and $U_{k|k} \triangleq \Var_{q_{xu}}[u_k]$ are required in the derivation of $q_\Lambda$ in Section \ref{sec:Filterq_L}, and they can be approximated using the recursive truncation algorithm. For the linear--Gaussian time update to be analytically tractable, the marginal distribution $q_{xu}(x_k)$ is approximated as a normal distribution
\begin{align}
q_{xu}(x_k) = \int q_{xu}(x_k,u_k) \,\mathrm{d} u_k \approx \N(x_{k|k},P_{k|k}) ,
\end{align}
where $P_{k|k} \triangleq \var_{q_{xu}}[x_k]$.

\subsection{ Derivations for $q_\Lambda$}
\label{sec:Filterq_L}
Using equation \eqref{eq:FilterIterativeOptimizationL} we obtain
\begin{align}
\log q_\Lambda(\Lambda_k) =& -\frac{1}{2}\E_{{q}_{xu}}[\tr\{(y_k-Cx_k-\Delta u_k)(y_k-Cx_k-\Delta u_k)^\t R^{-1}\Lambda_k\}]\nonumber\\
& - \frac{1}{2}\E_{{q}_{xu}}[\tr\{ u_k  u_k^\t{}\Lambda_k\}] + \sum_{i=1}^{n_y}\left( \frac{\nu_i}{2}\log[\Lambda_k]_{ii}- \frac{\nu_i}{2} [\Lambda_k]_{ii} \right)+c\\
=&-\frac{1}{2} \tr\left\{\left((y_k-Cx_{k|k}-\Delta u_{k|k})(y_k-Cx_{k|k}-\Delta u_{k|k})^\t+\begin{bmatrix}C&\Delta\end{bmatrix} Z_{k|k} \begin{bmatrix}C^\t\\\Delta^\t\end{bmatrix}\right) R^{-1} \Lambda_k\right\} \\
&-\frac{1}{2}\tr\left\{ (u_{k|k} u_{k|k}^\t+U_{k|k}) \Lambda_k \right\} + \sum_{i=1}^{n_y}\left( \frac{\nu_i}{2}\log[\Lambda_k]_{ii}- \frac{\nu_i}{2} [\Lambda_k]_{ii} \right)+c \\
=& \sum_{i=1}^{n_y} \left( \frac{\nu_i}{2} \log[\Lambda_k]_{ii} - \frac{\nu_i+[\Psi_k]_{ii}}{2}[\Lambda_k]_{ii} \right) +c 
\end{align}
where
\begin{equation}
\begin{split}
\Psi_k &=  (y_k-Cx_{k|k}-\Delta u_{k|k})(y_k-Cx_{k|k}-\Delta u_{k|k})^\t R^{-1} + \begin{bmatrix}C&\Delta\end{bmatrix} Z_{k|k} \begin{bmatrix}C^\t\\\Delta^\t\end{bmatrix} R^{-1} + u_{k|k}u_{k[k}^\t+U_{k|k}\label{eq:Filterpsi}
\end{split}
\end{equation}
and the moments $x_{k|k} \triangleq \E_{q_{xu}}[x_k]$, $u_{k|k} \triangleq \E_{q_{xu}}[u_k]$, $Z_{k|k} \triangleq \Var_{q_{xu}}\left[\begin{smallmatrix}x_k\\u_k\end{smallmatrix}\right]$, and $U_{k|k} \triangleq \Var_{q_{xu}}[u_k]$ are derived in Section \ref{sec:Filterq_xu} of this report. Therefore,
\begin{align}
q_\Lambda&(\Lambda_k)=\prod_{i=1}^{n_y}\G\left([\Lambda_k]_{ii};\frac{\nu_i}{2}+1,\frac{\nu_i+[\Psi_k]_{ii}}{2}\right).
\end{align}
Note that only the diagonal elements of the matrix $\Psi_k$ are required. In the derivations of Section \ref{sec:Filterq_xu} $\Lambda_{k|k} \triangleq \E_{q_\Lambda}[\Lambda_k]$ is required. $\E_{q_\Lambda}[\Lambda_k]$ is a diagonal matrix with the diagonal elements
\begin{align}
[\Lambda_{k|k}]_{ii}=\frac{\nu_i+2}{\nu_i+[\Psi_k]_{ii}}.
\end{align}


\vfill\clearpage
\bibliographystyle{IEEEtran}
\bibliography{IEEEabrv,VB-Skewness}

\end{document}